\documentclass[a4paper,UKenglish]{lipics-v2018}

\usepackage{microtype}

\usepackage{environments}
\input{enforcer.sty}
\input{monitor.sty}
\input{logic.sty}
\input{shorthand.sty}
\input{pic.sty}
\input{normalization.sty}

\usepackage{graphicx}
\usepackage{mathpartir}
\usepackage{semantic}
\usepackage{pifont}

\usepackage{color}

\usepackage{amsthm,thmtools}
\usepackage{cleveref}

\let\emptyset\varnothing

\theoremstyle{plain}
\newtheorem{proposition}[theorem]{Proposition}
\theoremstyle{definition}

\newtheoremstyle{break}{\topsep}{\topsep}{\itshape}{}{\bfseries}{.}{\newline}{}

\newcommand*{\dupcntr}[2]{%
	\expandafter\let\csname c@#1\expandafter\endcsname\csname c@#2\endcsname
}

\newcommand{\req}{\actN{req}\xspace}
\newcommand{\ans}{\actN{ans}\xspace}
\newcommand{\cls}{\actN{cls}\xspace}

\newcommand{\patReqA}{\actIn{\dvV}{\req}\xspace}
\newcommand{\patReqAB}{\actIn{(\dvV)}{\req}\xspace}
\newcommand{\patReqB}{\actIn{\dvV}{\req}\xspace}

\newcommand{\patAns}{\actOut{\dvV}{\ans}\xspace}

\newcommand{\symReqAB}{\actSN{\patReqAB}{\dvV{\neq}j}\xspace}
\newcommand{\symReqB}{\actSN{\patReqB}{\boolT}\xspace}

\newcommand{\symAns}{\actSN{\patAns}{\boolT}\xspace}

\newcommand{\actReq}{\actIn{i}{\req}\xspace}
\newcommand{\actAns}{\actOut{i}{\ans}\xspace}
\newcommand{\actCls}{\actIn{i}{\cls}\xspace}

\newcommand{\trnsReqAB}{\actSID{\patReqAB}{\dvV\neq j}\xspace}
\newcommand{\trnsReqB}{\actSN{\patReqB}{\actt}\xspace}

\newcommand{\trnsReqCB}{\actSTD{\patReqAB}{\dvV\neq j}\xspace}

\newcommand{\trnsAns}{\actSTN{\actOut{(\dvV)}{\ans}}{\boolT}{\actOut{\dvV}{\ans}}\xspace}

\newcommand{\procPVdef}{\rec{\rV}{\bigl(\ch{\prf{\actReq}{\prf{\actAns}\rV}}{\prf{\actCls}{\nil}}\bigr)}\xspace}
\newcommand{\procPVVdef}{\rec{\rV}{\bigl(\ch{\prf{\actReq}{\prf{\actAns}\rV}{\ch{\,}{\,\underline{\prf{\actReq}\rV}}}}{\prf{\actCls}{\nil}}\bigr)}\xspace}

\newcommand{\hVBdef}{\hMaxX{\hAnd{(\hNec{\symReqAB}\hNec{\symAns}\hVarX)\,}\,{(\hNec{\symReqAB}\hNec{\symReqB}\hFls)} }\xspace}

\newcommand{\hVdef}{\hMaxX{\hNec{\symReqAB}(\hAnd{\hNec{\symAns}\hVarX}{\hNec{\symReqB}\hFls})}\xspace}
\newcommand{\hVdefConc}{\hMaxX{\hVdefConcAP}\xspace}
\newcommand{\hVdefConcAP}{\hNec{\actReq}\;(\hAnd{\hNec{\actAns}\hVarX}{\hNec{\actReq}\hFls})}

\newcommand{\eVAdef}{\rec{\rV}{\bigl( \ch{\prf{\trnsReqCB}{\rV}}{\prf{\trnsAns}{\rV}} \bigr)}}

\newcommand{\pVg} {\ensuremath{\pV_\textbf{g}
	}\xspace}
\newcommand{\pVb} {\ensuremath{\pV_\textbf{b}
	}\xspace}

\newcommand{\eVr} {\ensuremath{\eV_\textbf{r}}\xspace}
\newcommand{\eVi} {\ensuremath{\eV_\textbf{i}}\xspace}
\newcommand{\eVs} {\ensuremath{\eV_\textbf{s}}\xspace}
\newcommand{\eVss} {\ensuremath{\eV_\textbf{t}}\xspace}

\newcommand{\eVb} {\ensuremath{\eV_\textbf{b}}\xspace}

\crefname{lemma}{lemma}{lemmas}
\Crefname{lemma}{Lemma}{Lemmas}
\crefname{proposition}{proposition}{propositions}
\Crefname{proposition}{Proposition}{Propositions}
\crefname{example}{example}{examples}
\Crefname{example}{Example}{Examples}

\title{On Runtime Enforcement via Suppressions}

\author{Luca Aceto \hfill Ian Cassar}{Gran Sasso Science Institute \hfill Reykjavik University \\ \& Reykjavik University \hfill  \& University of Malta \\{L'Aquila, Italy \& Reykjavik Iceland \hfill Reykjavik Iceland \& Msida, Malta }}{luca.aceto@gssi.it
\hfill ianc@ru.is}{}{}
\author{Adrian Francalanza \hfill Anna Ing\'{o}lfsd\'{o}ttir}{University of Malta \hfill Reykjavik University
	\\{Msida, Malta \hfill  Reykjavik, Iceland} }{adrian.francalanza@um.edu.mt \hfill annai@ru.is}{}{}

\authorrunning{L. Aceto, I. Cassar, A. Francalanza and  A. Ing\'{o}lfsd\'{o}ttir}

\Copyright{Ian Cassar, Adrian Francalanza, Luca Aceto and Anna Ing\'{o}lfsd\'{o}ttir}

\subjclass{\ccsdesc[500]{Theory of computation~Logic and verification}\;
	\ccsdesc[500]{Software and its engineering~Software verification}\;
	\ccsdesc[500]{Software and its engineering~Dynamic analysis}}

\keywords{Enforceability, Suppression Enforcement, Monitor Synthesis, Logic}

\category{}

\relatedversion{http://arxiv.org/abs/1804.08917}

\supplement{}

\acknowledgements{The research work disclosed in this publication is partially supported by the projects ``Developing Theoretical Foundations for Runtime Enforcement''
(184776-051)
and ``TheoFoMon: Theoretical Foundations for Monitorability''
(163406-051)
 of the Icelandic Research Fund, and by the Endeavour Scholarship Scheme (Malta),
part-financed by the
European Social Fund (ESF) - Operational Programme II – Cohesion Policy 2014-2020.}

\EventEditors{Sven Schewe and Lijun Zhang}
\EventNoEds{2}
\EventLongTitle{29th International Conference on Concurrency Theory (CONCUR 2018)}
\EventShortTitle{CONCUR 2018}
\EventAcronym{CONCUR}
\EventYear{2018}
\EventDate{September 4--7, 2018}
\EventLocation{Beijing, China}
\EventLogo{}
\SeriesVolume{118}
\ArticleNo{34} 

\nolinenumbers 

\begin{document}

\maketitle

\begin{abstract}
	 Runtime enforcement is a dynamic analysis technique that uses monitors to enforce the behaviour specified by some correctness property on an executing system.  The enforceability of a logic captures the extent to which the properties expressible via the logic can be enforced at runtime. We study the enforceability of Hennessy-Milner Logic with Recursion (\recHML) with respect to suppression enforcement.  We develop an operational framework for enforcement which we then use to formalise 
	 when a monitor enforces a \recHML property. 
	 We also show that the safety syntactic fragment of the logic, \SHML, is enforceable by providing an automated synthesis function that generates correct suppression monitors from \SHML formulas.
\end{abstract}

\section{Introduction} \label{sec:intro}

Runtime monitoring~\cite{francalanza2016theory,Cassar2017RV} is a dynamic analysis technique that is
becoming increasingly popular
in the turbid world of software development.
%
It uses code units called \emph{monitors} to aggregate system information, compare system execution against correctness specifications, or steer the execution of the observed system.
The technique has been used effectively to offload certain verification tasks to a post-deployment phase, thus complementing other (static) analysis techniques in  multi-pronged verification strategies---see \eg~\cite{RV:Test:2005,Bocchi2017,Jia2016Popl,Desai2017,Kejstova2017}.
\emph{Runtime enforcement} (RE)~\cite{Ligatti2005,Ligatti2010,Falcone2011} is a specialized monitoring technique, used to ensure that the behaviour of a system-under-scrutiny (\sus) is \emph{always} in agreement with some correctness specification.
It employs a specific kind of monitor (referred to as a
\emph{transducer}~\cite{Berstel:79:Transducer,Sakarovitch:2009:Transducers,
Alur:2011:Transducers} or an \emph{edit-automaton}~\cite{Ligatti2005,Ligatti2010}) to anticipate incorrect behaviour and counter it.
%
Such a monitor thus acts as a proxy between the \sus and the surrounding environment interacting with it, encapsulating the system to form a composite (monitored) system:
at runtime, the monitor  \emph{transforms} any incorrect executions exhibited by the \sus into correct ones by either \emph{suppressing},
 \emph{inserting} or \emph{replacing} events on behalf of the system.

We extend a recent line of research~\cite{Francalanza2017FMSD,Cassar2017RV,Achilleos2018FSTTCS,Achilleos2018Fossacs} and study RE approaches that adopt a \emph{separation of concerns} between the correctness specification, describing \emph{what} properties the \sus should satisfy, and the monitor, describing \emph{how} to  enforce these properties on the \sus.
Our work considers system properties expressed in terms of the process
logic \recHML~\cite{Kozen1983MuCalc,Larsen1990},
and explores what properties can be operationally enforced by monitors that can suppress system behaviour.
A central element for the realisation of such an approach is the \emph{synthesis} function: it automates the translation from the \emph{declarative} \recHML specifications to \emph{algorithmic} descriptions formulated as executable monitors.
%
Since
analysis tools ought to form part of the trusted computing base, enforcement monitoring should be, in and of itself,
correct.
However, it is unclear what is to be expected of the synthesised monitor to adequately enforce a \recHML formula.
Nor is it clear for which type of specifications should this approach be expected to work effectively---it has been well established that a number of properties are \emph{not} monitorable~\cite{chang1993,Pnueli2006,FraCini2015,Francalanza2017FMSD,Achilleos2018FSTTCS} and it is therefore reasonable to expect similar limits in the case of enforceability~\cite{Falcone2010}.
We therefore study the relationship between \recHML specifications and suppression monitors for enforcement, which allows us to address the above-mentioned concerns and make the following contributions:
\begin{description}
  \item[Modelling:] We develop a general framework for enforcement instrumentation that is parametrisable by any system behaviour that is expressed via labelled transitions, and can express suppression, insertion and replacement enforcement, \Cref{fig:mod-re}.
  \item[Correctness:] We give formal definitions for asserting when a monitor correctly enforces a formula defined over labelled transition systems, \Cref{def:enforceability,def:enforcement}.
  These definitions are parametrisable with respect to an instrumentation relation, an instance of which is our enforcement framework of \Cref{fig:mod-re}.
  \item[Expressiveness:] We provide enforceability results, \Cref{thm:strong-enf,thm:norm-equivalence} (but also \Cref{lemma:trace-transparency}),  by identifying a subset of \recHML formulas that can be (correctly) enforced by suppression monitors.
\end{description}
As a by-product of this study, we also develop a formally-proven correct synthesis function, \Cref{def:synthesis}, that then can be used for tool construction, along the lines of~\cite{Attard2016,Cassar2017Betty}.



The setup selected for our study serves a number of purposes.
For starters, the chosen logic, \recHML, is a branching-time logic that allows us to investigate enforceability for
properties describing computation graphs.
%
Second, the use of a highly expressive logic allows us to achieve a good degree of generality for our results, and so, by working in relation to logics like \recHML (a reformulation of the $\mu$-calculus), our work would also apply to other widely used logics (such as LTL and CTL~\cite{Clarke2008}) that are embedded within this logic. 
Third, since the logic is verification-technique agnostic, it fits better with the realities of software verification in the present world, where a \emph{variety} of techniques (\eg model-checking and testing) straddling both pre- and post-deployment phases are used.
In such cases, knowing which properties can be verified statically and which ones can be monitored for and enforced at runtime is crucial for devising effective multi-pronged verification strategies.
Equipped with such knowledge, one could also employ standard techniques \cite{Martinelli2005,Andersen1995,Lang2012TACAS} to decompose a non-enforceable property into a collection of smaller properties, a subset of which can then be enforced at runtime.

\noindent
\textbf{\textsf{Structure of the paper:}} \Cref{sec:prelim} revisits labelled transition systems and our touchstone logic, \recHML. The operational model for enforcement monitors and instrumentation is given in \Cref{sec:enf-model}. In \Cref{sec:enforceability} we formalise the interdependent notions of correct enforcement and enforceability. These act as a foundation for the development of a synthesis function in \Cref{sec:synthesis}, that produces \emph{correct-by-construction} monitors.
In \Cref{sec:strong-enforceability} we consider alternative definitions for enforceability for logics with a specific additional interpretation, and show that our proposed synthesis function is still correct with respect to the new definition.
\Cref{sec:conc} concludes and discusses related work.


\section{Preliminaries} \label{sec:prelim}

\noindent
\textbf{\textsf{The Model:}} We assume systems described as \emph{labelled transition systems} (LTSs), triples $\langle\Sys,\Act\cup\sset{\actt},\rightarrow\rangle$ consisting of a set of \emph{system states}, $\pV,\pVV,\pVVV\in\Sys$, a set of \emph{observable actions}, $\acta,\actb\in\Act$, and a distinguished silent action $\actt\notin\Act$ (where $\actu \in \Act\cup\sset{\actt}$), and a \emph{transition} relation,
$\reduc\;\subseteq(\Sys\times\Act \cup \sset{\actt}\times\Sys)$.
We write $\pV\traS{\actu}\pVV$ in lieu of $(\pV,\actu,\pVV) \in\, \rightarrow$, and use
$\pV \wtraS{\actu} \pV'$ to denote weak transitions representing
$\pV (\traS{\actt})^{\ast}\cdot\traS{\actu}\cdot(\traS{\actt})^{\ast} \pV'$.
We refer to $\pV'$ as a \actu-derivative of $\pV$.
Traces, $\tr,\trr\in\Act^\ast$ range over (finite) sequences of observable actions, and we write $\pV \wtraS{\tr} \pVV$ to denote a sequence of weak transitions  $\pV \wtraS{\acta_1} \ldots \wtraS{\acta_n} \pVV$ for $\tr = \acta_1,\ldots,\acta_n$.
We also assume the classic notion of \emph{strong bisimilarity}~\cite{Milner1992CCS,Sangiorgi2011Bisim} for our model, $\pV \bisim \pVV$, using it as our touchstone system equivalence.  The syntax of the regular fragment of CCS~\cite{Milner1992CCS} is occasionally used to concisely describe LTSs in our examples.

\begin{figure}[t]
	\textbf{\textsf{Syntax}}
	\vspace{-1mm}
	\!\!{\small$$\begin{array}{r@{\,}llc@{\,}llc@{\,}ll}
	\hV,\hVV \in \recHML \bnfdef&  \hTru &(\text{truth}) & \bnfsepp  &\hFls \; &(\text{falsehood})& \bnfsepp & \hBigOr{i\in\IndSet}{\hV_i} &(\text{disjunction}) \\[1mm]
	\bnfsepp& \hBigAndU{i\in\IndSet}{\hV_i} & (\text{conjunction}) & \bnfsepp & \hSuf{\actSN{\pate}{\bV}}{\hV} & (\text{possibility}) &\bnfsepp& \hNec{\actSN{\pate}{\bV}}{\hV} & (\text{necessity}) \\[1mm]
	\bnfsepp & \hMinXF & (\text{least fp.}) &\bnfsepp& \hMaxXF & (\text{greatest fp.}) &\bnfsepp& \hVarX & (\text{fp. variable})
\end{array}$$} \vspace{-2mm}

	\textbf{\textsf{Semantics}}
	{
	{\small	\[
	  \begin{array}{r@{\,}c@{\;}l@{\qquad}r@{\,}c@{\;}l@{\qquad}r@{\,}c@{\;}l}
		\hSemS{\hTru,\rho}  &\defEquals& \Sys &
		\hSemS{\hFls,\rho}  &\defEquals& \emptyset
		&\hSemS{\hVarX,\rho} &\defEquals& \rho(\hVarX)
		\\[1mm]
		\hSemS{\hBigAnd{i\in\IndSet}{\hV_i},\rho} & \defEquals & \bigintersectU{i\in\IndSet}\hSemS{\hV_i,\rho} &
		\!\hSemS{\hMaxXF,\rho} & \defEquals & \bigcup \Set{S \;|\;  S \subseteq \hSemS{\hV,\rho[\hVarX\mapsto S]}}
		\\[1mm]
		\hSemS{\hBigOr{i\in\IndSet}{\hV_i},\rho} & \defEquals & \bigunionU{i\in\IndSet}\hSemS{\hV_i,\rho} &
		\!\hSemS{\,\hMinXF,\rho} & \defEquals & \bigcap \Set{S \;|\;  \hSemS{\hV,\rho[\hVarX\mapsto S]} \subseteq S\,}
		\\[1mm]
		\hSemS{\,\hNec{\actSN{\pate}{\bV}}{\hV},\rho}  & \defEquals&
			\multicolumn{7}{l}{
			\!\!\!\Set{\pV \;|\;  (\forall\acta,\pVV\cdot \pV  \wtra{\acta} \pVV  \text{ and } (\exists \sV \cdot\mtch{\pate}{\acta}{=}\sV \text{ and } \ceval{\bV\sV}{\boolT})) \text{ implies } q \in \hSemS{\hV\sV,\rho}\!}}
			\\[1mm]
		\hSemS{\hSuf{\actSN{\pate}{\bV}}{\hV},\rho}  & \defEquals&
			\multicolumn{7}{l}{
			\!\!\!\Set{\pV \;|\; \exists\acta,\pVV,\sV\cdot(\pV  \wtra{\acta} \pVV   \text{ and } \mtch{\pate}{\acta}{=}\sV \text{ and } \ceval{\bV\sV}{\boolT} \text{ and } q \in \hSemS{\hV\sV,\rho}) \!}
		}
	\end{array}
	\]}
	}\vspace{-3mm}
	\caption{\recHML Syntax and Semantics}
	\label{fig:recHML} \vspace{-4mm}
\end{figure}

\medskip
\noindent
\textbf{\textsf{The Logic:}} We consider a slightly generalised version of \recHML \cite{Larsen1990,Aceto2007Book} that uses \emph{symbolic actions} of the form \actSN{\pate}{\bV}.
%
\emph{Patterns},
\pate, abstract over actions
using \emph{data variables} $\dvV,\dvVV,\dvVVV\in\Var$; in a pattern, they may either occur free, \dvV, or as binders, $(\dvV)$ where a \emph{closed pattern} is one without any free variables.
We assume a (partial) \emph{matching function} for \emph{closed} patterns $\mtch{\pate}{\acta}$ that returns a substitution \sV (when successful) mapping variables in \pate to the corresponding values in
\acta, \ie if we instantiate every bound variable \dvV in \pate with $\sV(\dvV)$ we obtain \acta.
The \emph{filtering condition}, \bV,
contains
variables found in \pate 
and evaluates \wrt the substitutions returned by successful matches.
Put differently, a \emph{closed} symbolic action \actSN{\pate}{\bV} is one where \pate is closed and $\fv{\bV} \subseteq \bv{\pate}$; it
denotes the \emph{set} of actions $\hSemS{\actSN{\pate}{\bV}}\defeq\Setdef{\acta}{\exists\sV\cdot\mtch{\pate}{\acta}\!=\!\sV \textsl{ and } \ceval{\bV\sV}{\boolT}}$ and allows more adequate reasoning about LTSs with infinite actions (\eg actions carrying data from infinite domains).

The logic syntax is given in \Cref{fig:recHML} and assumes a countable set of logical variables $\hVarX,\hVarY\!\in\!\LVars$.
Apart from standard logical constructs such as conjunctions and disjunctions ($\hBigAnd{i\in\IndSet}{\hV_i}$ describes a \emph{compound} conjunction,  $\hAnd{{\hV_1}}{\hAnd{\ldots}{{\hV_n}}}$,
where
$\IndSet=\sset{1,..,n}$  is a finite set of indices, and similarly for disjunctions), and the characteristic greatest and least fixpoints (\hMaxXF and \hMinXF bind free occurrences of \hVarX in \hV), the logic uses necessity and possibility modal operators with symbolic actions, \hNec{\actSN{\pate}{\bV}}{\hV} and \hSuf{\actSN{\pate}{\bV}}{\hV}, where \bv{\pate} bind free data variables in \bV and \hV.
Formulas in \recHML are interpreted over 
the system powerset domain where $S{\in}\pset{\Sys}$.
The semantic definition of \Cref{fig:recHML}, \hSemS{\hV,\rho},  is given for  \emph{both} open and closed formulas. It employs a valuation  from logical variables to sets of states, $\rho\in(\LVars \rightarrow \pset{\Sys})$, which permits an inductive definition on the structure of the formulas;
$\rho'=\rho[\hVarX\mapsto S]$ denotes a valuation where $\rho'(\hVarX) = S$ and  $\rho'(\hVarY) = \rho(\hVarY)$ for all other $\hVarY\neq \hVarX$.
The only non-standard cases are those for the modal formulas, due to the use of symbolic actions.
Note
that we recover the standard logic for symbolic actions \actSN{\pate}{\bV} whose pattern \pate does not contain variables ($\pate {=} \acta$ for some \acta)  and whose condition holds trivially ($\bV {=} \boolT$); in such cases we write $\hNec{\acta}{\hV}$ and $\hSuf{\acta}{\hV}$ for short.
We generally assume \emph{closed} formulas, \ie without free logical and data variables, and write \hSemS{\hV} in lieu of  \hSemS{\hV,\rho} since the interpretation of a closed \hV is independent of
$\rho$.
A system \pV \emph{satisfies} formula \hV whenever $\pV{\,\in\,}\hSemS{\hV}$ whereas a formula \hV is \emph{satisfiable},  $\hV\in\Sat$, whenever there exists a system \pVV such that $\pVV\in\hSemS{\hV}$.

\begin{example} \label[example]{ex:uhml-formula} Consider two systems (a good system, \pVg, and a bad one, \pVb) implementing
	a server that interacts on port $i$, repeatedly accepting \emph{requests} that are \emph{answered} by outputting on the same port, and terminating the service once a \emph{close} request is accepted  (on the same port).
	Whereas $\pVg$ outputs an answer (\actAns) for every request (\actReq),
	$\pVb$
	occasionally refuses to answer a given request (see the underlined branch). Both systems terminate with \actCls.
\begin{displaymath}
	\pVg=\procPVdef \qquad\quad \pVb=\procPVVdef
\end{displaymath}
We can specify that two consecutive requests on port $i$ indicate invalid behaviour via the \recHML formula $ \hV_0{\defeq}\hVdefConc$;
it defines an invariant property (\hMaxB{\hVarX}{\ldots}) requiring that whenever a system interacting on
$i$ inputs a request, it cannot input a subsequent request, \ie $\hNec{\actReq}\hFls$, unless it outputs an answer beforehand, in which case the formula recurses, \ie $\hNec{\actAns}\hVarX$.
Using symbolic actions, we can generalise $\hV_0$ by requiring the property to hold for \emph{any} interaction happening on \emph{any} port number \emph{except}  $j$.
\begin{align*}
	\hV_1 & \defeq\hVdef
\end{align*}
In $\hV_1$, \patReqAB binds the free occurrences of \dvV found in $\dvV{\neq}j$ and \hAnd{\hNec{\symAns}\hVarX}{\hNec{\symReqB}\hFls}.
Using \Cref{fig:recHML}, one can check that $\pVg {\in} \hSemS{\hV_1}$, whereas $\pVb {\not\in} \hSemS{\hV_1}$ since
$\pVb\traS{\actReq}\cdot\traS{\actReq}\ldots$  \qed
\end{example}

\section{An Operational Model for Enforcement} \label{sec:enf-model}

\begin{figure}[t]
	\noindent\textbf{\textsf{Syntax}}
	\begin{align*}
		\eV,\eVV\in\Trn
		&\bnfdef \quad \eIden \qquad
		\bnfsepp  \eTrns{\pate}{c}{\pate'}{\eV} \qquad
		\bnfsepp   \chBigText{i\in\IndSet}\, \eV_i \qquad
		\bnfsepp  \rec{\rV}{\eV} \qquad
		\bnfsepp  \rV
	\end{align*}
	\noindent\textbf{\textsf{Dynamics}}\vspace{-2mm}
	\begin{mathpar}
		\inference[\rtit{eId}]{ }{\eIden \traS{\ioact{\actu}{\actu}} \eIden }
		\and
		\inference[\rtit{eSel}]{\eV_j \traSS{\actgu} \eVV_j}{
		\chBigText{i\in\IndSet}\,\eV_i \traSS{\actgu} \eVV_j}[$j{\in}\IndSet$]
		\and
		\inference[\rtit{eRec}]{\eV\sub{\rec{\rV}{\eV}}{\rV} \tra{\actgu} \eVV }{\rec{\rV}{\eV} \tra{\actgu} \eVV}
		\and
		\inference[\rtit{eTrn}]{
		\mtch{\pate}{\actg} = \sV
		&&
    \ceval{\bV\sV}{\boolT}
		&&
		\actu{\,=\,}\pate'\sV
		}{\eTrns{\pate}{c}{\pate'}{\eV} \traS{\ioact{\actg}{\actu}} \eV\sV}
		\vspace{-5mm}
	\end{mathpar} 
	\noindent\textbf{\textsf{Instrumentation}} 	\vspace{-2mm}
	\begin{mathpar}
		\inference[\rtit{iTrn}]
		{\pV \tra{\acta}\pV' \\ \eV\tra{\ioact{\acta}{\actu}}\eVV}
		{ \eI{\eV}{\pV} \tra{\actu} \eI{\eVV}{\pV'}}
		\quad
		\inference[\rtit{iAsy}]
		{ \pV \tra{\actt}\pV' }
		{ \eI{\eV}{\pV} \tra{\actt} \eI{\eV}{\pV'} }
    \quad
		\inference[\rtit{iIns}]
		{\eV \tra{\ioact{\actdot}{\actu}}\eVV}
    {\eI{\eV}{\pV} \tra{\actu} \eI{\eVV}{\pV}}
		\quad
		\inference[\rtit{iTer}]
		{\pV \tra{\acta}\pV'
		\\ \eV\ntra{\acta}
		&& \eV\ntra{\actdot}
		}
		{ \eI{\eV}{\pV} \tra{\acta} \eI{\eIden}{\pV'}}
	\end{mathpar}
	\caption[]{A model for transducers (\IndSet is a finite index set and $\protect\eV{\!\!\ntraSS{\actg}}$ means $\protect\nexists\actu,\eVV\cdot\eV{\traS{\actgu}}\eVV$)}
	\label{fig:mod-re}
\end{figure}

Our operational mechanism for enforcing properties over systems uses the (symbolic) transducers $\eV,\eVV\in\Trn$ defined in \Cref{fig:mod-re}.
The transition rules in \Cref{fig:mod-re} assume closed terms, \ie for every \emph{symbolic-prefix transducer}, \prf{\actSTN{\pate}{\bV}{\pate'}}{\eV}, \pate is closed and
$\bigl(\fv{\bV} {\cup} \fv{\pate'} {\cup} \fv{\eV}\bigr) \subseteq \bv{\pate}$, and yield an LTS with labels of the form \ioact{\actg}{\actu}, where $\actg\in(\Act\,{\cup}\sset{
\actdot})$.
Our syntax assumes a well-formedness constraint where for every \prf{\actSTN{\pate}{\bV}{\pate'}}{\eV}, $\bv{\bV} {\cup} \bv{\pate'} = \emptyset$.
Intuitively, a transition $\eV \tra{\ioact{\acta}{\actu}} \eVV$
denotes the fact that the transducer in state \eV \emph{transforms} the visible action $\acta$ (produced by the system) into the action $\actu$ (which can possibly become silent) and transitions into state \eVV.
In this sense, the transducer action \ioact{\acta}{\actt} represents the \emph{suppression} of
action \acta,  action \ioact{\acta}{\actb} represents the \emph{replacing} of \acta by \actb, and \ioact{\acta}{\acta} denotes the
\emph{identity} transformation.
The special case \ioact{\actdot}{\acta} encodes the \emph{insertion} of \acta, where \actdot represents
that the transition is not induced by any system action.

The key transition rule in \Cref{fig:mod-re} is \rtit{eTrn}.
It states that the symbolic-prefix transducer \prf{\actSTN{\pate}{\bV}{\pate'}}{\eV} can transform an (extended) action \actg into the concrete action \actu, as long as  the action matches with pattern \pate with substitution \sV, $\mtch{\pate}{\actg}{=}\sV$, and the condition is satisfied by \sV, \ceval{\bV\sV}{\boolT}
(the matching function is lifted to extended actions and patterns in the obvious way, where $\mtch{\actdot}{\actdot}{=}\emptyset$).
In such a case, the transformed action is
$\actu{=}\pate'\sV$, \ie the action \actu resulting from the instantiation of the free data variables in pattern $\pate'$ with the corresponding values mapped by \sV,  and the transducer state reached is ${\eV}\sV$.
%
%
By contrast, in rule \rtit{eId}, the transducer \eIden acts as the identity and leaves actions unchanged.
The remaining rules are fairly standard and unremarkable.

\Cref{fig:mod-re} also describes an \emph{instrumentation} relation which relates the behaviour of the \sus \pV with the transformations of a transducer monitor \eV that
\emph{agrees} with the (observable) actions \Act of \pV.
The term \eI{\eV}{\pV} thus denotes the resulting \emph{monitored system} whose behaviour is defined in terms of $\Act{\,\cup}\sset{\actt}$ from the system's LTS.
Concretely, rule \rtit{iTrn} states that when a system \pV transitions with an observable action \acta  to $\pV'$ and the transducer \eV can \emph{transform} this action into \actu and transition to $\eVV$, the instrumented system \eI{\eV}{\pV} transitions with action \actu to \eI{\eVV}{\pV'}.
However, when \pV transitions with a silent action, rules \rtit{iAsy} allows it to do so independently of the transducer.
Dually, rule \rtit{iIns} allows the transducer to \emph{insert} an action \actu independently of \pV's behaviour.
Rule \rtit{iTer} is analogous to standard monitor instrumentation rules  for premature termination of the transducer~\cite{francalanza2016theory,Francalanza2017FMSD,Fra17:Concur,Achilleos2018Fossacs}, and accounts for underspecification of transformations.
Thus, if a system \pV transitions with an observable action \acta to $\pV'$, and the transducer \eV does not specify how to transform it ($\eV\ntra{\acta}$), nor can it transition to a new transducer state by inserting an action ($\eV\ntra{\actdot}$), the system is still allowed to transition while the transducer's transformation activity is ceased, \ie it acts like the identity \eIden from that point onwards.

\begin{example} \label[example]{ex:transducers}
	Consider the insertion transducer \eVi and the replacement transducer \eVr below:
	\begin{align*}
		\eVi & \defeq \prf{
				\actSTN{\actdot}{\btrue}{\actReq}
			}
			{
				\prf{
					\actSTN{\actdot}{\btrue}{\actAns}
				}{\eIden}
			}
		\\
		\eVr &\defeq \rec{\rV}{\bigl(
       \ch{
			 		\ch{
			 			\prf{\actSTN{\actIn{(\dvV)}{\req}}{\btrue}{\actIn{j}{\req}}}{\rV}
					}{
						\prf{\actSTN{\actOut{(\dvV)}{\ans}}{\btrue}{\actOut{j}{\ans}}}{\rV}
					}
			 }{
			 	  \prf{\actSTN{\actIn{(\dvV)}{\cls}}{\btrue}{\actIn{j}{\cls}}}{\rV}
			 }
		 \bigr)}
	\end{align*}
	When instrumented with a system, \eVi inserts the two successive actions \actReq and \actAns before behaving as the identity.
	Concretely in the case of \pVb we can only start the computation as:
	\begin{equation*}
		\eI{\eVi}{\pVb}
		\traSS{\actReq}
		\eI{\prf{\actSTN{\actdot}{\btrue}{\actAns}}{\eIden}}{\pVb}
		\traSS{\actAns}
		\eI{\eIden}{\pVb}
		\traSS{\acta}
		\ldots 
		\qquad (\text{where }\pVb\traS{\acta})
	\end{equation*}
	By contrast, \eVr transforms input actions  with either payload $\req$ or  \cls and output actions with payload $\ans$ on any port name, into the respective actions on port $j$.
	For instance:
	\begin{equation*}
		\eI{\eVr}{\pVb}
		\traSS{\actIn{j}{\req}}
		\eI{\eVr}{\prf{\actAns}{\pVb}}
		\traSS{\actOut{j}{\ans}}
		 \eI{\eVr}{\pVb}
		 \traSS{\actIn{j}{\cls}}
		 \eI{\eVr}{\nil}
	\end{equation*}
	Consider now the two suppression transducers \eVs and \eVss for actions on ports other than $j$:
	\begin{align*}
		\eVs & \defeq \eVAdef
		\\
		\eVss & \defeq \rec{\rV}{\bigl(
			  \prf{\actSTN{\patReqAB}{\dvV\neq j}{\patReqA}}{
					\rec{\rVV}{\bigl(
					\ch{\prf{\actSTN{\patAns}{\btrue}{\patAns}}{\rV}}
					{\prf{\actSTD{\patReqA}{\btrue}}{\rVV}} \bigr)}
				}
		\bigr)}
	\end{align*}
	Monitor \eVs suppresses any requests on ports other than $j$, and continues to do so after any answers on such ports.
	When instrumented with \pVb, we can observe the following behaviour:
	\begin{equation*}
		\eI{\eVs}{\pVb}
		\traSS{\actt}
		\eI{\eVs}{\prf{\actAns}{\pVb}}
		\traSS{\actAns}
		\eI{\eVs}{\pVb}
		\traSS{\actt}
		\eI{\eVs}{\prf{\actAns}{\pVb}}
		\traSS{\actAns}
		\eI{\eVs}{\pVb} \ldots
	\end{equation*}
	Note that \eVs does not specify a transformation behaviour for when the monitored system produces inputs with payload other than \req\!\!\!\!.
	The instrumentation handles this underspecification by ceasing suppression activity;
	in the case of \pVb we get
	\begin{math}
		\eI{\eVs}{\pVb}
		\traSS{\actCls}
		\eI{\eIden}{\nil}
	\end{math}.
The transducer \eVss performs slightly more elaborate transformations.
For interactions on ports other than $j$, it suppresses consecutive input requests following any serviced request (\ie an input on \req followed by an output on \ans) sequence.
For \pVb we can observe the following:
\begin{align}
	\eI{\eVss}{\pVb}
	&
	\traSS{\actReq}
	\eI{\rec{\rVV}{\bigl( \ch{\prf{\actSTN{\actAns}{\btrue}{\actAns}}{\eVss}}{\prf{\actSTD{\actReq}{\btrue}}{\rVV}} \bigr)}}{\pVb}
	\nonumber
	\\
	&
	\traSS{\actt}
	\eI{
	\rec{\rVV}{\bigl( \ch{\prf{\actSTN{\actAns}{\btrue}{\actAns}}{\eVss}}{\prf{\actSTD{\actReq}{\btrue}}{\rVV}} \bigr)}
	}{
		\prf{\actAns}{\pVb}
	}
	\traSS{\actAns}
	\eI{\eVss}{\pVb} \tag*{\qed}
\end{align}
\end{example}

In the sequel, we find it convenient to refer to \underline{\pate} as the transformed pattern \pate where all the binding occurrences $(\dvV)$ are converted to free occurrences $\dvV$.
As shorthand notation, we elide the second pattern $\pate'$ in a transducer \prf{\actSTN{\pate}{\bV}{\pate'}}{\eV} whenever $\pate'{=}\underline{\pate}$ and simply write \prf{\actSID{\pate}{\bV}}{\eV}; note that if $\bv{\pate}=\emptyset$, then $\underline{\pate} {=}\pate$.
Similarly, we elide \bV whenever $\bV{=}\btrue$.
This allows us to express \eVss from \Cref{ex:transducers} as
\begin{math}
	\rec{\rV}{\bigl(
			\prf{\actSID{\actIn{(\dvV)}{\req}}{\dvV{\neq}j}}{
				\rec{\rVV}{\bigl(
					\ch{
						\prf{\actSTRID{\actOut{\dvV}{\ans}}}{\rV}
					}{
						\prf{\actSTR{\actIn{\dvV}{\req}}{\actt}}{\rVV}
					}
					\bigr)}
			}
	\bigr)}
\end{math}.
%

\section{Enforceability} \label{sec:enforceability}

The \emph{enforceability} of a logic rests on the relationship between the semantic behaviour specified by the logic on the one hand, and the ability of the operational mechanism (the transducers and instrumentation of \Cref{sec:enf-model} in our case) to enforce the specified behaviour
on the other.


\begin{definition}[Enforceability] \label[definition]{def:enforceability}
	A logic \LSet is enforceable iff \emph{every} formula $\hV{\in}\LSet$ is \emph{enforceable}.
	 A formula \hV is \emph{enforceable} iff there \emph{exists} a transducer \eV such that \eV \emph{enforces} \hV. \qed
\end{definition}

\Cref{def:enforceability} depends on what is considered to be an adequate definition for ``\eV \emph{enforces} \hV''.
It is reasonable to expect that the latter definition should concern \emph{any} system that the transducer \eV---hereafter referred to as the \emph{enforcer}---is instrumented with.
%
In particular,
for \emph{any} system \pV, the resulting composite system obtained from instrumenting the  enforcer \eV with it should satisfy the property of interest, \hV, whenever this property \emph{is satisfiable}.

\begin{definition}[Sound Enforcement] \label[definition]{def:senf}  Enforcer \eV \emph{soundly enforces} a formula \hV, denoted as \senfdef{\eV}{\hV}, iff for \emph{all} $ \pV\in\Sys$, $\hV \in Sat$ implies  $\eI{\eV}{\pV}\in\hSemS{\hV}$ holds.\qed
\end{definition}

\begin{example}\label[example]{ex:sound-enf} Recall $\hV_1$, \pVg and \pVb from \Cref{ex:uhml-formula} where $\pVg \in \hSemS{\hV_1}$ (hence $\hV_1\in\Sat$) and $\pVb \not\in\hSemS{\hV_1}$.
For the enforcers \eVi, \eVr, \eVs and \eVss presented in \Cref{ex:transducers}, we have:
\begin{itemize}
	\item $\eI{\eVi}{\pVb}\not\in\hSemS{\hV_1}$, since $\eI{\eVi}{\pVb}\traS{\actReq}\cdot\traS{\actAns}\eI{\eIden}{\pVb}\traS{\actReq}\eI{\eIden}{\pVb}\traS{\actReq}\eI{\eIden}{\pVb}$.
	This counter example implies that $\neg\senfdef{\eVi}{\hV_1}$.
	%
	\item $\eI{\eVr}{\pVg}\in\hSemS{\hV_1}$ and $\eI{\eVr}{\pVb}\in\hSemS{\hV_1}$.
	 Intuitively, this is because the ensuing instrumented systems only generate  (replaced) actions that are not of concern to
	 $\hV_1$.
	Since this behaviour applies to any system \eVr is composed with, we can conclude that \senfdef{\eVr}{\hV_1}.
\item $\eI{\eVs}{\pVg}\in\hSemS{\hV_1}$ and $\eI{\eVs}{\pVb}\in\hSemS{\hV_1}$
 because the resulting instrumented systems never produce inputs with $\req$ on a port number other than $j$.
 We can thus conclude that \senfdef{\eVs}{\hV_1}.
\item $\eI{\eVss}{\pVg}\in\hSemS{\hV_1}$ and $\eI{\eVss}{\pVb}\in\hSemS{\hV_1}$.
Since the resulting instrumentation suppresses consecutive input requests  (if any) after any number of serviced requests on any port other than $j$, we can conclude that \senfdef{\eVss}{\hV_1}. \qed
\end{itemize}
\end{example}

By some measures, sound enforcement is a relatively weak requirement for adequate enforcement as it does not regulate the \emph{extent} of the induced enforcement.
%
More concretely, consider the case of enforcer \eVs from \Cref{ex:transducers}.
Although \eVs manages to suppress the violating executions of system $\pVb$, thereby bringing it in line with property $\hV_1$, it needlessly modifies the behaviour of $\pVg$ (namely it prohibits it from producing any inputs with \req\! on port numbers that are not $j$), even though it satisfies $\hV_1$.
Thus, in addition to sound enforcement we require a \emph{transparency} condition for adequate enforcement.
The requirement dictates that whenever a system $\pV$ already satisfies the property \hV, the assigned enforcer \eV should not alter the behaviour of $\pV$.
Put differently, the behaviour of the enforced system should be behaviourally equivalent to the original system.

\begin{definition}[Transparent Enforcement] \label[definition]{def:tenf} An enforcer \eV is \emph{transparent} when enforcing a formula \hV, denoted as \tenfdef{\eV}{\hV}, iff  for \emph{all} $\pV\in\Sys$, $\pV\in\hSemS{\hV}$  implies $\eI{\eV}{\pV}\sim\pV$. \qed
\end{definition}

\begin{example} \label[example]{ex:transparency}
	We have already argued---via the counter example \pVg---why \eVs does \emph{not} transparently enforce $\hV_1$.
	We can also argue easily why $\neg\tenfdef{\eVr}{\hV_1}$ either:
	the simple system $\prf{\actReq}{\nil}$ trivially satisfies $\hV_1$ but, clearly, we have the inequality $\eI{\eVr}{\prf{\actReq}{\nil}} \not\sim \prf{\actReq}{\nil}$ since
	$\eI{\eVr}{\prf{\actReq}{\nil}} \traS{\actIn{j}{\req}} \eI{\eVr}{\nil}$
	and $\prf{\actReq}{\nil} \centernot{\traS{\actIn{j}{\req}}}$.

	It turns out that enforcer \tenfdef{\eVss}{\hV_1}, however.  Although this property is not as easy to show---due to the universal quantification over all systems---we can get a fairly good intuition for why this is the case via the example \pVg: it satisfies $\hV_1$ and
	$\eI{\eVss}{\pVg} \sim \pVg$ holds. \qed
\end{example}

\begin{definition}[Enforcement]
	\label[definition]{def:enforcement}
A monitor \eV enforces property \hV whenever it does so $(i)$ soundly, \Cref{def:senf} and $(ii)$ transparently, \Cref{def:tenf}. \qed
\end{definition}

\newcommand{\hVns}{\ensuremath{\hV_{\textsf{ns}}}\xspace}
\newcommand{\pVA}{\ensuremath{\pV_{\textsf{a}}}\xspace}
\newcommand{\pVRA}{\ensuremath{\pV_{\textsf{ra}}}\xspace}
\newcommand{\pVR}{\ensuremath{\pV_{\textsf{r}}}\xspace}
\newcommand{\eVA}{\ensuremath{\eV_{\textsf{a}}}\xspace}
\newcommand{\eVR}{\ensuremath{\eV_{\textsf{r}}}\xspace}

For any reasonably expressive logic (such as \recHML), it is usually the case that \emph{not} every formula can be enforced, as the following example informally illustrates.

\begin{example} \label[example]{ex:shml-only} Consider the \recHML property \hVns, together with the two systems \pVRA and \pVR:
	\begin{align*}
		\hVns\defeq\hOr{\hNec{\actReq}\hFls\;}{\;\hNec{\actAns}\hFls} &&
		\pVRA \defeq \ch{\prf{\actReq}{\nil}}{\prf{\actAns}{\nil}} &&
		\pVR \defeq \prf{\actReq}{\nil}
	\end{align*}
  A system satisfies \hVns if \emph{either} it cannot produce action \actReq  \emph{or} it cannot produce action \actAns.
	Clearly, \pVRA violates this property as it can produce both.
	This system can only be enforced via action suppressions or replacements because insertions would immediately break transparency.
	Without loss of generality, assume that our monitors employ suppressions (the same argument applies for action replacement).
	The monitor $\eVR \defeq \rec{\rVV}{\bigl( \ch{\prf{\actSN{\actReq}{\actt}}{\rVV}}{\prf{\actSN{\actAns}{\actt}}{\rVV}} \bigr)}$ would in fact be able to suppress the offending actions produced by \pVRA, thus obtaining $\eI{\eVR}{\pVRA} \in \hSemS{\hVns}$.
	However, it would also suppress the sole action \actReq produced by the system \pVR, even though this system satisfies \hVns.
	This would, in turn, violate the transparency criterion of \Cref{def:tenf} since it needlessly suppresses \pVR's actions, \ie although $\pVR \in \hSemS{\hVns}$ we have $\eI{\eVR}{\pVR} \not\sim \pVR$.
	The intuitive reason for this problem is that a monitor cannot, in principle, look into the computation graph of a system, but is limited to the behaviour the system exhibits at runtime.  \qed
	%
\end{example}

\section{Synthesising Suppression Enforcers} \label{sec:synthesis}

Despite their merits, \Cref{def:enforcement,def:enforceability} are not easy to work with.
The universal quantifications over all systems in \Cref{def:senf,def:tenf} make it hard to establish that a monitor correctly enforces a
property.
Moreover, according to \Cref{def:enforceability}, in order to determine whether a particular property is enforceable or not, one would need to show the existence of a monitor that correctly enforces it;
put differently, showing that a property is \emph{not} enforceable entails another universal quantification, this time showing that no monitor can possibly enforce the property.
Lifting the question of enforceability to the level of a (sub)logic entails a further universal quantification, this time on all the logical formulas of the logic; this is often an infinite set.
%
We address these problems in two ways.
First, we identify a non-trivial syntactic subset of \recHML that is \emph{guaranteed to be enforceable}; in a multi-pronged approach to system verification, this could act as a guide for whether the property should be considered at a pre-deployment or post-deployment phase.
Second, for \emph{every} formula \hV in this enforceable subset, we provide an \emph{automated procedure} to \emph{synthesise} a monitor \eV from it that correctly enforces \hV when instrumented over arbitrary systems, according to \Cref{def:enforcement}.
This procedure can then be used as a basis for constructing tools that automate property enforcement. 

\begin{figure}[t]
	\begin{align*}
		\hV,\hVV\in\SHML&\;\bnfdef\;
		\hTru \bnfseppp
		\hFls \bnfseppp
		\hBigAndU{i\in\IndSet}{\hV_i}  \bnfseppp
		\hNec{\actSN{\pate}{\bV}}{\hV} \bnfseppp
		\hVarX  \bnfseppp
		\hMaxXF
	\end{align*} \vspace{-6mm}
	\caption{The syntax for the safety \recHML fragment, \SHML.}
	\label{fig:shml-syn} \vspace{-4mm}
\end{figure}

In this paper, we limit our enforceability study to suppression monitors, transducers that are only allowed to intervene by dropping (observable) actions.
Despite being more constrained, suppression monitors
side-step problems associated with what data to use in a payload-carrying action generated by the enforcer,
as in the case of insertion and replacement monitors:
the notion of a default value for certain data domains is not always immediate.
Moreover, suppression monitors are particularly useful for enforcing \emph{safety} properties, as shown in \cite{Ligatti2005,Bielova2011PhD,Falcone2012}.
Intuitively, a suppression monitor would suppress actions as soon as it becomes apparent that a violation is about to be committed by the \sus.
Such an intervention intrinsically relies on the \emph{detection} of a violation.
To this effect, we use a prior result from \cite{Francalanza2017FMSD}, which identified a maximally-expressive logical fragment of \recHML that can be handled by violation-detecting (recogniser) monitors.
We thus limit our enforceability study to this maximal safety fragment, called \SHML, since a \emph{transparent} suppression monitor cannot judiciously suppress actions without first detecting a (potential) violation.
%
%
\Cref{fig:shml-syn} recalls the syntax for \SHML. The logic is restricted to \emph{truth} and \emph{falsehood} (\hTru and \hFls), conjunctions (\hBigAnd{i\in\IndSet}{\hV}), and necessity modalities (\hNec{\actSN{\pate}{\bV}}{\hV}), while recursion may only be expressed through greatest fixpoints (\hMaxXF);
the semantics
follows
that of \Cref{fig:recHML}.

%

A
standard way how to achieve our aims would be to $(i)$ define a (total) synthesis function $\eSem{-}:: \SHML \mapsto \Trn$ from \SHML formulas to suppression monitors and $(ii)$ then show that for \emph{any} $\hV\in\SHML$, the synthesised monitor $\eSem{\hV}$ enforces \hV.
Moreover, we would also require the synthesis function to be compositional, whereby the definition of the enforcer for a composite formula is defined in terms of the enforcers obtained for the constituent subformulas.  There are a number of reasons for this requirement.
For one, it would simplify our analysis of the produced monitors and allow us to use standard inductive proof techniques to prove properties about the synthesis function, such as the aforementioned criteria $(ii)$.  However, a naive approach to such a scheme is bound to fail, as discussed in the next example.

\begin{example} \label[example]{ex:naive-compositional}  Consider a semantically equivalent reformulation of $\hV_1$ from \Cref{ex:uhml-formula}.
	\begin{align*}
		\hV_2 &\defeq\; \hVBdef
	\end{align*}
   At an intuitive level, the suppression monitor that one would expect to obtain for the subformula $\hV'_2\defeq\hNec{\symReqAB}\hNec{\symReqB}\hFls$ is
	 $\prf{\trnsReqAB}{\rec{\rVV}{\prf{\trnsReqB}{\rVV}}}$ (\ie an enforcer that repeatedly drops any \req inputs following a \req input on the same port), whereas the monitor obtained for the subformula $\hV''_2\defeq\hNec{\symReqAB}\hNec{\symAns}\hVarX$ is
	 $\prf{\trnsReqAB}{ \prf{\actSTRID{\actOut{\dvV}{\ans}}}{{\rV}} }$
	 (assuming some variable mapping from
	 \hVarX to
	 \rV).
	 These monitors would then be combined in the synthesis for $\hMaxX{\hAnd{\hV''_2}{\hV'_2}}$ as
	 \begin{align*}
 		 \eVb &\defeq\; \rec{\rV}{\ch{\bigl(\prf{\trnsReqAB}{\prf{\actSTRID{\actOut{\dvV}{\ans}}}{{\rV}}} \bigr) \,}{\,\bigl(\prf{\trnsReqAB}{\rec{\rVV}{\prf{\trnsReqB}{\rVV}}} \bigr)}}
 	\end{align*}
	One can easily see that \eVb does \emph{not} behave deterministically, \emph{nor} does it soundly enforce $\hV_2$.  For instance, for the violating system $\prf{\actReq}{\prf{\actReq}{\nil}} \not\in\hSemS{\hV_2}(=\hSemS{\hV_1})$ we can observe the transition sequence
	\begin{math}
		\eI{\eVb}{\prf{\actReq}{\prf{\actReq}{\nil}}} \traSS{\actReq} \eI{\prf{\actSTRID{\actOut{i}{\ans}}}{{\eVb}}}{\prf{\actReq}{\nil}} \traSS{\actReq}
		\eI{\eIden}{\nil}
	\end{math}. \qed
\end{example}

Instead of complicating our synthesis function to cater for anomalies such as those presented in \Cref{ex:naive-compositional}---also making it \emph{less} compositional in the process---we opted for a two stage synthesis procedure.
First, we consider a \emph{normalised} subset for \SHML formulas which is amenable to a (straightforward)  synthesis function definition that is compositional.
This also facilitates the proofs for the conditions required by \Cref{def:enforcement} for any synthesised enforcer.
Second, we show that every \SHML formula can be reformulated in this normalised form without affecting its semantic meaning.
We can then show that our two-stage approach is expressive enough to show the enforceability for all of \SHML.

\begin{definition}[\SHML normal form]
	The set of normalised \SHML formulas is defined as:
	\begin{align*}
		\hV,\hVV\in\SHMLnf\;\bnfdef\; \hTru
		\bnfseppp\hFls
		\bnfseppp \hBigAndU{i\in\IndSet}\hNec{\actSN{\pate_i}{\bV_i}}{\hV_{i}}
		\bnfseppp\hVarX
		\bnfseppp\hMaxXF\;.
	\end{align*}
	The above grammar combines necessity operators with conjunctions  into one construct $\hBigAndU{i\in\IndSet}\hNec{\actSN{\pate_i}{\bV_i}}{\hV_{i}}$.  Normalised \SHML formulas are required to satisfy two further conditions:
	\begin{enumerate}
		\item For every  $\hBigAndU{i\in\IndSet}\hNec{\actSN{\pate_i}{\bV_i}}{\hV_{i}}$, for all $j, h \in \IndSet$ where $j{\neq}h$ we have $\hSemS{\actSN{\pate_j}{\bV_j}} \cap \hSemS{\actSN{\pate_h}{\bV_h}} = \emptyset$.
		\item For every \hMaxXF we have $\hVarX \in \fv{\hV}$. \qed
	\end{enumerate}
\end{definition}

In a (closed) normalised \SHML formula, the basic terms \hTru  and \hFls can never appear unguarded unless they are at the top level (\eg we can never have $\hAnd{\hV}{\hFls}$ or $\hMax{\hVarX_{0}}{\ldots\hMax{\hVarX_{n}}{\hFls}}$).
 Moreover, in any conjunction of necessity subformulas, $ \hBigAndU{i\in\IndSet}\hNec{\actSN{\pate_i}{\bV_i}}{\hV_{i}}$, the
necessity guards are \emph{disjoint} and \emph{at most one} necessity guard can satisfy any particular action.


\begin{definition} \label[definition]{def:synthesis}
	The synthesis function $\eSem{-}:\SHMLnf{\,\mapsto\,}\Trn$ is defined inductively as:
	\begin{align*}
	\eSem{\hVarX} &\defeq \rV
	\qquad\qquad
	\eSem{\hTru} \defeq \eSem{\hFls} \defeq \eIden
	\qquad\qquad
	\eSem{\hMaxX{\varphi}} \defeq \rec{\rV}{\eSem{\varphi}}\\
	\eSem{\hBigAndD{i{\,\in\,}\IndSet}\hNec{\actSN{\pate_i}{\bV_i}}{\varphi_{i}}}
	&\defeq
	\rec{\rVV}{\chBigI
	  \begin{xbrace}{ll}
		\prf{\actSTN{\pate_i}{\bV_i}{\actt}}{\rVV} & \qquad \text{if } \hV_i{=}\hFls\\
		\prf{\actSTN{\pate_i}{\bV_i}{\underline{\pate_i}}}{\eSem{\varphi_i}} & \qquad \text{otherwise}
		\end{xbrace}
		}\tag*{
			\qed
		}
	\end{align*}
\end{definition}

The synthesis function is compositional.
It assumes a bijective mapping between formula variables and monitor recursion variables and converts logical variables \hVarX accordingly,  whereas maximal fixpoints, $\hMaxXF$, are converted into the corresponding recursive enforcer.
The synthesis also converts truth and falsehood formulas, \hTru and \hFls, into the identity enforcer \eIden.
Normalized conjunctions, $\hBigAndU{i{\,\in\,}\IndSet}\hNec{\actSN{\pate_i}{\bV_i}}{\hV_{i}}$, are synthesised into a \emph{recursive summation} of enforcers, \ie $\rec{\rVV}{\eV_i}$, where $\rVV$ is fresh, and every branch $\eV_i$ can be either of the following:
\begin{enumerate}[$(i)$]
	\item when $\eV_i$ is derived from a branch of the form $\hNec{\actSN{\pate_i}{\bV_i}}\hV_i$ where $\hV_i{\neq}\hFls$, the synthesis produces an enforcer with the \emph{identity transformation} prefix, $\actSTN{\pate_i}{\bV_i}{\underline{\pate_i}}$, followed by the enforcer synthesised from the continuation $\hV_{i}$,
	\ie  $\hNec{\actSN{\pate_i}{\bV_i}}\hV_i$ is synthesised as $\prf{\actSTN{\pate_i}{\bV_i}{\underline{\pate_i}}}{\eSem{\hV_i}}$;
	\item when $\eV_i$ is derived from a branch of the form $\hNec{\actSN{\pate_i}{\bV_i}}\hFls$, the synthesis produces a \emph{suppression transformation}, $\actSTN{\pate_i}{\bV_i}{\actt}$, that drops every concrete action matching the symbolic action \actSN{\pate_i}{\bV_i}, followed by the recursive variable of the branch \rVV,
	\ie a branch of the form $\hNec{\actSN{\pate_i}{\bV_i}}\hFls$ is translated into $\prf{\actSTN{\pate_i}{\bV_i}{\actt}}{\rVV}$.
\end{enumerate}

\begin{example} \label[example]{ex:synthesis}
	Recall formula $\hV_1$ from \Cref{ex:uhml-formula}, recast in term of \SHMLnf's grammar:
	\begin{align*}
		\qquad\hV_1 & \defeq \hMaxX{\hBigAndD{}\bigl(\,\hNec{\symReqAB}\;\bigl(\hAnd{\hNec{\symAns}\hVarX\;}{\;\hNec{\symReqB}\hFls}\bigr)\bigr)}
	\end{align*}
	Using the synthesis function defined in \Cref{def:synthesis}, we can generate the enforcer
	\begin{align*}
 \eSem{\hV_1} & = \rec{\rV}{\rec{\rVVV}{\chBig{}\bigl(\,\prf{\symReqAB}{\rec{\rVV}{ (\ch{\prf{\symAns}{\rV}\,}{\,\prf{\actSTD{\patReqB}{\boolT}}{\rVV}})}}\bigr)}}
 \intertext{which can be optimized by removing redundant recursive constructs (\eg \rec{\rVVV}{\_}), obtaining:}
 & =
 \rec{\rV}{\prf{\symReqAB}{\rec{\rVV}{(\ch{\prf{\symAns}{\rV}\,}{\,\prf{\actSTD{\patReqB}{\boolT}}{\rVV}})}}} \; =\; \eVss
 \end{align*} \\[-12mm]\qed
\end{example}

\noindent We now present the first main result to the paper.

\begin{theorem}[Enforcement] \label[theorem]{thm:strong-enf} The (sub)logic $\SHMLnf$ is enforceable.
\end{theorem}
\begin{proof} By \Cref{def:enforceability}, the
	result follows if we show that for all $\hV{\,\in\,}\SHMLnf$,  $ \eSem{\hV}\;\textsl{enforces}\; \hV$.  By \Cref{def:enforcement}, this is a corollary following from \Cref{lemma:soundness,lemma:transparency} stated below. 
\end{proof}

\begin{proposition}[Enforcement Soundness] \label[proposition]{lemma:soundness} For every system $\pV{\,\in\,}\Sys$ and $\hV{\,\in\,}\SHMLnf$ then 
	$\hV \in \Sat \imp \eI{\eSem{\hV}{}}{\pV}{\,\in\,}\hSemS{\hV} $. \qed
\end{proposition}
\begin{proposition}[Enforcement Transparency] \label[proposition]{lemma:transparency} For every system $\pV{\,\in\,}\Sys$ and  $\hV{\,\in\,}\SHMLnf$ then 
	$\pV{\,\in\,}\hSemS{\hV} \imp  \eI{\eSem{\hV}{}}{\pV}\bisim\pV$. \qed
\end{proposition}

\medskip

Following \Cref{thm:strong-enf}, to show that \SHML is an enforceable logic, we only need to show that for every $\hV\in\SHML$ there exists a corresponding $\hVV\in\SHMLnf$ with the same semantic meaning, \ie $\hSemS{\hV} = \hSemS{\hVV}$.
In fact, we go a step further and provide a constructive proof using a transformation $\nSem{-}:\SHML \mapsto \SHMLnf$ that derives a semantically equivalent \SHMLnf formula from a standard \SHML formula.
As a result, from an arbitrary \SHML formula $\hV$ we can then automatically synthesise a correct enforcer using $\eSem{\nSem{\hV}}$ which is useful for tool construction.

Our transformation $\nSem{\hV}$ relies on a number of steps; here we provide an outline of these steps.
First, we assume \SHML formulas that only use symbolic actions with \emph{normalised} patterns \pate, \ie patterns that do not use any data or free data variables (but they may use bound data variables).
In fact, any symbolic action \actSN{\pate}{\bV} can be easily converted into a corresponding one using normalised patterns as shown in the next example.

\begin{example}
	Consider the symbolic action \actSN{\actOut{\dvV}{\ans}}{\dvV \neq j}.
	It may be converted to a corresponding normalised symbolic action by replacing every occurrence of a data or free data variable in the pattern by a fresh bound variable, and then add an equality constraint between the fresh variable and the data or data variable it replaces in the pattern condition.
	In our case, we would obtain \actSN{\actOut{(\dvVV)}{(\dvVVV)}}{\dvV{\neq}j \wedge \dvVV{=}\dvV \wedge \dvVVV{=}\ans}. \qed
\end{example}

Our algorithm for converting \SHML formulas (with normalised patterns) to \SHMLnf formulas, $\nSem{-}$, is based on Rabinovich's work \cite{Rabinovich1993} for determinising systems of equations which, in turn relies on the standard powerset construction for converting NFAs into DFAs. It consists in the following six stages that we outline below:
\begin{enumerate}
	\item We unfold each recursive construct in the formula, to push recursive definitions inside the formula body. \Eg the formula
	\hMaxX{\bigl(\hAnd{\hNec{\actSN{\pate_1}{\bV_1}}{\hVarX}}{\hNec{\actSN{\pate_2}{\bV_2}}{\hFls}}\bigr)} is expanded to the formula
	\hAnd{\hNec{\actSN{\pate_1}{\bV_1}}{ \bigl(\hMaxX{\hAnd{\hNec{\actSN{\pate_1}{\bV_1}}{\hVarX}}{\hNec{\actSN{\pate_2}{\bV_2}}{\hFls}}} \bigr)}}{\hNec{\actSN{\pate_2}{\bV_2}}{\hFls}}.
	\item The formula is converted into a system of equations.  \Eg the expanded formula from the previous stage is converted into the set $\sset{X_0 = \hAnd{\hNec{\actSN{\pate_1}{\bV_1}}{X_0}}{\hNec{\actSN{\pate_2}{\bV_2}}{X_1}}, X_1 = \hFls}$.
	\item For every equation, the symbolic actions in the right hand side that are of the same kind are alpha-converted so that their bound variables match. \Eg Consider $X_0 = \hAnd{\hNec{\actSN{\pate_1}{\bV_1}}{X_0}}{\hNec{\actSN{\pate_2}{\bV_2}}{X_1}}$ from the previous stage where, for the sake of the example,
	$\pate_1 = \actIn{(\dvV_1)}{(\dvV_2)}$ and $\pate_2 = \actIn{(\dvV_3)}{(\dvV_4)}$.  The patterns in the symbolic actions are made syntactically equivalent by renaming $\dvV_3$ and $\dvV_4$ in \actSN{\pate_2}{\bV_2} into  $\dvV_1$ and $\dvV_2$ respectively.
	\item For equations with matching patterns in the symbolic actions, we create a variant that symbolically covers all the (satisfiable) permutations on the symbolic action conditions. \Eg Consider $X_0 = \hAnd{\hNec{\actSN{\pate_1}{\bV_1}}{X_0}}{\hNec{\actSN{\pate_1}{\bV_3}}{X_1}}$ from the previous stage.  We expand this to
	$X_0 = \hNec{\actSN{\pate_1}{\bV_1\wedge \bV_3}}{X_0}\wedge\hNec{\actSN{\pate_1}{\bV_1\wedge \bV_3}}{X_1} \wedge
	\hNec{\actSN{\pate_1}{\bV_1\wedge \neg(\bV_3)}}{X_0}\wedge\hNec{\actSN{\pate_1}{\neg(\bV_1)\wedge \bV_3}}{X_1}$.
	\item For equations with branches having \emph{syntactically equivalent} symbolic actions, we carry out a unification procedure akin to standard powerset constructions. \Eg we convert the equation from the previous step to
	$X_{\sset{0}} = \hNec{\actSN{\pate_1}{\bV_1\wedge \bV_3}}{X_{\sset{0,1}}}\wedge
	\hNec{\actSN{\pate_1}{\bV_1\wedge \neg(\bV_3)}}{X_{\sset{0}}}\wedge\hNec{\actSN{\pate_1}{\neg(\bV_1)\wedge \bV_3}}{X_{\sset{1}}}$
	using the (unified) fresh variables $X_{\sset{0}}, X_{\sset{1}}$ and $X_{\sset{0,1}}$.
	\item From the unified set of equations we generate again the \SHML formula starting from  $X_{\sset{0}}$.  This procedure may generate redundant recursion binders, \ie \hMaxXF where $X \not\in \fv{\hV}$, and we filter these out in a subsequent pass.
\end{enumerate}

We now state the second main result of the paper.

\begin{theorem}[Normalisation] \label{thm:norm-equivalence}
	For any  $\hV{\in}\SHML$ there exists $\hVV{\in}\SHMLnf $ s.t. $ \hSemS{\hV}{=}\hSemS{\hVV}$.
\end{theorem}
\begin{proof}
	The witness formula in normal form is \nSem{\hV}, where we show that each and every stage in the translation procedure preserves semantic equivalence. 
\end{proof}

\section{Alternative Transparency Enforcement} \label{sec:strong-enforceability}

Transparency for a property \hV, \Cref{def:tenf}, only restricts enforcers from modifying the behaviour of satisfying systems, \ie when $\pV{\in}\hSemS{\hV}$, but fails to specify any enforcement behaviour for the cases when the \sus violates the property $\pV{\notin}\hSemS{\hV}$.   In this section, we consider an alternative transparency requirement for a property \hV that incorporates the expected enforcement behaviour for \emph{both} satisfying and violating systems. More concretely, in the case of safety languages such as \SHML, a system typically violates a property along a specific set of execution traces; in the case of a satisfying system this set of ``violating traces'' is \emph{empty}.   However, not every behaviour of a violating system would be part of this set of violating traces  and, in such cases, the respective enforcer should be required to leave the generated behaviour unaffected.

\begin{definition}[Violating-Trace Semantics]
  \label[definition]{def:violation-trace-sem}
	A logic $\LSet$  with an interpretation over systems $\hSemS{-}:\LSet \mapsto \pset{\Sys}$ has a violating-trace semantics whenever it has a secondary interpretation $\hSemS{-}_v:\LSet \mapsto \pset{\Sys\times\Act^\ast}$ satisfying the following conditions for all $\hV\in\LSet$:
\begin{enumerate}
	\item $(\pV,\tr) \in \hSemS{\hV}_v$ implies $\pV \notin \hSemS{\hV}$ and $\pV \wtraS{\tr}$\;,
	\item $\pV \notin \hSemS{\hV}$ implies $\exists \tr \cdot (\pV,\tr) \in \hSemS{\hV}_v$\;.  \qed
\end{enumerate}
\end{definition}

We adapt the work in \cite{FraSey2015} to give \SHML  a violating-trace semantics.  Intuitively,  the judgement $(\pV,\tr) \in \hSemS{\hV}_v$
according to \Cref{def:violation-trace-sem-shml} below, denotes the fact that \pV violates the \SHML property \hV along trace \tr.

\begin{definition}[Alternative Semantics for \SHML \cite{FraSey2015}]
	 \label[definition]{def:violation-trace-sem-shml}
	 The forcing relation $\vsatL \subseteq \bigl(\Sys\times\Act^\ast\times\SHML\bigr)$ is the least relation satisfying the following rules:
	 \begin{align*}
 		 (\pV,\epsilon,\hFls)\in\R&
		&& \quad \text{ always }\\
 		(\pV,\tr,\hBigAndU{i\in\IndSet}\hV_i)\in\R&
		&& \quad \text{ if } \exists j\in\IndSet \text{ such that } (\pV,\tr,\hV_j)\in\R  \\
 		 (\pV,\acta\tr,\hNec{\actSN{\pate}{\bV}}\hV)\in\R&
		&& \quad \text{ if } \mtch{\pate}{\acta}{=}\sV, \ceval{\bV\sV}{\boolT} \text{ and } \pV\wtraS{\acta}\pV' \text{ and } (\pV',\tr,\hV\sV)\in\R \\
 		 (\pV,\tr,\hMaxXF)\in\R&
		&& \quad \text{ if } (\pV,\tr,\hV\sub{\hMaxXF}{\hVarX})\in\R\;.
 	\end{align*}
	We write \vsat{\pV}{\tr}{\hV} (or $(\pV,\tr) \in \hSemS{\hV}_v$) in lieu of $(\pV,\tr,\hV) \in\, \vsatL$.
	We say that trace \tr is a \emph{violating trace} for \pV with respect to \hV whenever \vsat{\pV}{\tr}{\hV}.
	Dually, \tr is a \emph{non-violating trace} for \hV whenever there does \emph{not} exist a system \pV such that \vsat{\pV}{\tr}{\hV}. \qed
\end{definition}

\begin{example} \label[example]{ex:alternative-trans}
	Recall $\hV_1, \pVb$ from \Cref{ex:uhml-formula} where $\hV_1 \in \SHML$, and also $\eVss$ from \Cref{ex:sound-enf} where we argued in \Cref{ex:synthesis} that $\eSem{\hV_1} = \eVss$ (modulo cosmetic optimisations).  Even though $\pVb \not\in \hSemS{\hV_1}$,
	not all of its exhibited behaviours constitute violating traces: for instance, $\pVb \wtraS{\actReq\cdot\actAns}\pVb$ is not a violating trace according to \Cref{def:violation-trace-sem-shml}.  Correspondingly, we also have $\eI{\eVss}{\pVb} \wtraS{\actReq\cdot\actAns} \eI{\eVss}{\pVb}$. \qed
\end{example}

\begin{theorem}[Adapted and extended from \cite{FraSey2015}]
	The alternative interpretation $\hSemS{-}_v$ of \Cref{def:violation-trace-sem-shml} is a violating-trace semantics for \SHML (with $\hSemS{-}$ from \Cref{fig:recHML}) in the sense of \Cref{def:violation-trace-sem}. \qed
\end{theorem}

Equipped with \Cref{def:violation-trace-sem-shml} we can define an alternative definition for transparency that concerns itself with preserving exhibited traces that are non-violating.  We can then show that the monitor synthesis for \SHML of \Cref{def:synthesis} observes non-violating trace transparency.



\begin{definition}[Non-Violating Trace Transparency] \label[definition]{def:vtenf}
	An enforcer \eV is \emph{transparent} with respect to the non-violating traces of a formula \hV, denoted as \nvtenfdef{\eV}{\hV}, iff  for \emph{all} $\pV\in\Sys$ and $\tr\in\Act^{*}$, when $\nvsat{\pV}{\tr}{\hV}$ then
	\begin{itemize}
		\item $\pV\wtraS{\tr}\pV' \imp \eI{\eV}{\pV}\wtraS{\tr}\eI{\eV'}{\pV'}$ for some $\eV'$, and
		\item $\eI{\eV}{\pV}\wtraS{\tr}\eI{\eV'}{\pV'}  \imp  \pV\wtraS{\tr}\pV'$
		. \qed
	\end{itemize} 
\end{definition}

\begin{proposition}[Non-Violating Trace Transparency] \label[proposition]{lemma:trace-transparency}
	  For all $\hV\in \SHML$, $\pV\in\Sys$ and $\tr\in\Act^{*}$, when $\nvsat{\pV}{\tr}{\hV}$ then
	  \begin{itemize}
	  	\item $\pV\wtraS{\tr}\pV' \imp \eI{\eSem{\hV}{}}{\pV}\wtraS{\tr}\eI{\eV'}{\pV'}$, and
	  	\item $\eI{\eSem{\hV}{}}{\pV}\wtraS{\tr}\eI{\eV'}{\pV'}  \imp  \pV\wtraS{\tr}\pV'$. \qed
	  \end{itemize}
\end{proposition}

We can thus obtain a new definition for ``\eV \emph{enforces} \hV'' instead of \Cref{def:enforcement} by requiring sound enforcement, \Cref{def:tenf}, and non-violating trace transparency, \Cref{def:vtenf} (instead of the transparent enforcement of \Cref{def:tenf}).  This in turn gives us a new definition for enforceability for a logic, akin to \Cref{def:enforceability}. Using \Cref{lemma:soundness,lemma:trace-transparency}, one can show that \SHML is also enforceable with respect to the new definition as well.

\section{Conclusion} \label{sec:conc}

This paper presents a preliminary investigation of the enforceability of properties expressed in a process logic.
We have focussed on a highly expressive and standard logic, \recHML,
 and studied the ability to enforce \recHML properties via a specific kind of monitor that performs suppression-based enforcement.
We concluded that \SHML, identified in earlier work as a maximally expressive safety fragment of \recHML, is also an enforceable logic.
To  show this, we first defined enforceability for logics and system descriptions interpreted over labelled transition systems.
Although enforceability builds upon soundness and transparency requirements that have been considered in other work, our branching-time framework allowed us to consider novel definitions for these requirements.
We also contend that the definitions that we develop for the enforcement framework are fairly modular:
\eg the instrumentation relation is independent of the specific language constructs defining our transducer monitors and it functions as expected as long as the transition semantics of the transducer and the system are in agreement.
Based on this notion of enforcement, we devise a two-phase procedure to synthesise correct enforcement monitors.
We first identify a syntactic subset of our target logic \SHML that affords  certain structural properties and permits a compositional definition of the synthesis function.
We then show that, by augmenting existing rewriting techniques to our setting, we can convert any \SHML formula into this syntactic subset.


\paragraph*{Related Work}



In his seminal work \cite{schneider2000}, Schneider regards a property (in a linear-time setting) to be enforceable if its \emph{violation} can be \emph{detected} by a \emph{truncation automaton}, and prevents its occurrence via system termination; by preventing misbehaviour, these enforcers can only enforce safety properties.  Ligatti \etal in \cite{Ligatti2005} extended this work via \emph{edit automata}---an enforcement mechanism capable of \emph{suppressing} and \emph{inserting} system actions. A property is thus enforceable if it can be expressed as an edit automaton that \emph{transforms} invalid executions into valid ones via suppressions and insertions.
Edit automata are capable of enforcing instances of safety and liveness properties, along with other properties such as infinite renewal properties \cite{Ligatti2005,Bielova2011PhD}. 
As a means to assess the correctness of these automata, the authors introduced \emph{soundness} and \emph{transparency}.
%
%
In both of these settings, there is no clear separation between the specification and the enforcement mechanism, and properties are encoded in terms of the languages accepted by the enforcement model itself, \ie as edit/truncation automata. By contrast, we keep the specification and verification aspects of the logic separate.

Bielova \etal \cite{Bielova2011PhD,Bielova2011Predictability} remark that soundness and transparency do not specify to what extent  a transducer should modify an invalid execution.
They thus introduce a \emph{predictability} criterion to prevent transducers from transforming invalid executions arbitrarily.
More concretely, a transducer is \emph{predictable} if one can predict the number of transformations that it will apply in order to transform an invalid execution into a valid one, thereby preventing enforcers from applying unnecessary transformations over an invalid execution.
Using this notion, Bielova \etal thus devise a more stringent notion of enforceability.
Although we do not explore this avenue, \Cref{def:vtenf} may be viewed as an attempt to constrain transformations of violating systems in a branching-time setup, and should be complementary to these predictability requirements.
%

K{\"o}nighofer \etal in \cite{Konighofer2017} present a synthesis algorithm that produces action replacement transducers called \emph{shields} from safety properties encoded as automata-based specifications. Shields analyse the inputs and outputs of a reactive systems and enforce properties by modifying the least amount of output actions whenever the system deviates from the specified behaviour. By definition, shields should adhere to two desired properties, namely correctness and minimum deviation which are, in some sense, analogous to soundness and transparency respectively. 
Falcone \etal in \cite{Falcone2010,Falcone2011,Falcone2012}, also propose synthesis procedures to translate properties $-$ expressed as Streett automata $-$ into the \resp enforcers.
%
%
The authors show that most of the property classes defined within the \emph{Safety-Progress hierarchy} \cite{Pnueli1990} are enforceable, as they can be encoded as Streett automata and subsequently converted into enforcement automata.
As opposed to Ligatti \etal, both K{\"o}nighofer \etal and Falcone \etal separate the specification of the property from the enforcement mechanism, but unlike our work they do not study the enforceability of a branching time logic.
%
%
%
%

To the best of our knowledge, the only other work that tackles enforceability for the modal $\mu$-calculus \cite{Kozen1983MuCalc} (a reformulation of \recHML) is that of Martinelli \etal in \cite{martinelli2006,martinelli2007}.
Their approach is, however, different from ours.
In addition to the $\mu$-calculus formula to enforce, their synthesis function also takes a ``witness'' system satisfying the formula as a parameter.
This witness system is then used as the behaviour that is mimicked by the instrumentation via suppression, insertion or replacement mechanisms.
Although the authors do not explore automated correctness criteria such as the ones we study in this work, it would be interesting to explore the applicability of our methods to their setting.

%
%

Bocchi \etal \cite{Bocchi2017} adopt \emph{multi-party session types} to project the global protocol specifications of distributed networks to \emph{local types} defining a local protocol for every process in the network that are then either verified statically via typechecking or enforced dynamically via suppression monitors.
%
%
To implement this enforcement strategy, the authors define a dynamic monitoring semantics for the local types that suppress process interactions so as to conform to the assigned local specification.
They prove local soundness and transparency for monitored processes that, in turn, imply global soundness and transparency by construction.
%
Their local enforcement is closely related to the suppression enforcement studied in our work with the following key differences:
 $(i)$ well-formed branches in a session type are, by construction, \emph{explicitly disjoint} via the use of distinct choice labels (\ie similar to our normalised subset \SHMLnf), whereas we can synthesise enforcers for \emph{every} \SHML formula using a normalisation procedure;
 $(ii)$ they give an LTS semantics to their local specifications (which are session types) which allows them to state that a process satisfies a specification when its behaviour is bisimilar to the operational semantics of the local specification---we do not change the semantics of our formulas, which is left in its original denotational form;
$(iii)$ they do not provide transparency guarantees for processes that violate a specification, along the lines of \Cref{def:vtenf};
%
$(iv)$ Our monitor descriptions sit at a lower level of abstraction than theirs using a dedicated language, whereas theirs have a session-type syntax with an LTS semantics (\eg repeated suppressions have to be encoded in our case using the recursion construct while this is handled by their high-level instrumentation semantics).


In \cite{Castellani2016}, Castellani \etal  adopt session types to define reading and writing privileges amongst processes in a network as global types for information flow purposes.
These global types are projected into local monitors capable of preventing read and write violations by adapting certain aspects of the network.
Although their work is pitched towards adaptation \cite{Cassar2017RV,Cassar2016IFM}, rather than enforcement, in certain instances they adapt the network by suppressing messages or by replacing messages with messages carrying a default nonce value.
It would be worthwhile investigating whether our monitor correctness criteria could be adapted or extended to this information-flow setting.

\paragraph*{Future Work}
We plan to extend this work along two different avenues.  On the one hand, we will attempt to extend the enforceable fragment of  \recHML.
For a start, we intend to investigate maximality results for suppression monitors, along the lines of \cite{Francalanza2017FMSD,Achilleos2018FSTTCS}.
We also plan to consider more expressive enforcement mechanisms such as insertion and replacement actions.
Finally, we will also investigate more elaborate instrumentation setups, such as the ones explored in \cite{Achilleos2018Fossacs}, that can reveal refusals in addition to the actions performed by the system.

On the other hand, we also plan to study the implementability and feasibility of our framework.
We will consider target languages for our monitor descriptions that are closer to an actual implementation (\eg an actor-based language along the lines of \cite{FraSey2015}).
We could then employ refinement analysis techniques and use our existing monitor descriptions as the abstract specifications that are refined by the concrete monitor descriptions.
The more concrete synthesis can then be used for the construction of tools that are more amenable towards showing correctness guarantees.





\bibliography{refs}
\newpage

\appendix

\section{Proving Enforcement Correctness}
\label{sec:app:correctness-proofs}
In this section we present proofs ascertaining the correctness of our enforcers. We prove \Cref{thm:strong-enf}, by proving that the enforcers synthesised by our synthesis function are \emph{sound} and \emph{transparent}. We prove these two criteria in \Cref{sec:proof-soundness,sec:proof-transparency}. Finally, we prove that our synthesised enforcers also abide by \emph{non-violating trace transparency} in \Cref{sec:proof-trace-transparency}.

In order to facilitate our proofs we also use an alternative satisfaction semantics for \SHML as explained below.
\paragraph*{Alternative \SHML Semantics} An alternative semantics for \SHML was presented by Aceto \etal in \cite{Aceto1999TestingHML,Aceto2007Book} in terms of a \emph{satisfaction relation}, \hSat. When restricted to \SHML, \hSat is the \emph{largest relation} \R satisfying the implications defined in \Cref{fig:uhml-sat}. 
\begin{figure}[h]
	\begin{displaymath}
	\begin{array}{r@{\;\,}c@{\;\,}ll}
	(\pV,\hTru)&\in&\R & \imp \textsl{ true } \\[.5mm]
	(\pV,\hFls)&\in&\R & \imp \textsl{ false } \\[.5mm]
	(\pV,\hBigAnd{i\in\IndSet}\hV_i)&\in&\R & \imp (\pV,\hV_i)\in\R \textsl{ for all } i{\,\in\,}\IndSet \\[.5mm]
	(\pV,\hNec{\actS}\hV)&\in&\R &\imp (\forall\acta,\pVV\cdot\pV\wtraS{\acta}\pVV \textsl{ and } \mtchS{\actS}{\acta}=\sV)\,\imp\, (\pVV,\hV\sV)\in\R \\[.5mm]
	(\pV,\hMaxXF)&\in&\R & \imp (\pV,\hV\Sub{\hMaxXF}{\hVarX})\in\R \\[.5mm]
	\end{array}
	\end{displaymath}
	\caption{A Satisfaction relation for \SHML formulas}
	\label{fig:uhml-sat}
\end{figure}

The satisfaction relation states that truth, \hTru, is \emph{always satisfied}, while falsehood, \hFls, can \emph{never be satisfied}. Conjunctions, $\hBigAnd{i\in\IndSet}\hV_i$ are satisfied when \emph{all branches} are satisfied (\ie $\forall i{\,\in\,}\IndSet \text{ such that } \pV\hSat\hV_i$), while 
necessities, $\hNec{\actS}\hV$, are satisfied by a process \pV when \emph{all derivatives} \pVV that are reachable over an action \acta where $\mtchS{\actS}{\acta}\!=\!\sV$ (possibly none), also satisfy $\hV\sV$, \ie $\pVV{\,\hSat\,}\hV\sV$. 
Finally, a process \pV satisfies a maximal fixpoint \hMaxXF when it is also able to satisfy an \emph{unfolded version} of \hV, \ie $\pV\hSat\hV\sub{\hMaxX{\hV}}{\hVarX}$. 

The satisfaction semantics, $\pV\hSat\hV$, agrees with the denotational semantics of the \SHML subset of \recHML, \hSemS{\hV}, presented in \Cref{fig:recHML}, so that $\pV\hSat\hV$ can be used in lieu of $\pV\in\hSemS{\hV}$ (see \cite{Aceto1999TestingHML,Aceto2007Book} for more detail).

\subsection{Proving Soundness} \label{sec:proof-soundness}
	$$ \forall\pV{\,\in\,}\Sys,\hV\in\SHMLnf\;\cdot\;\hV{\in}\Sat \imp \eI{\eSem{\hV}}{\pV}{\,\hSat\,}\hV $$
	To prove this lemma we must show that relation \R (below) is a \emph{satisfaction relation} (\hSat) as defined by the rules in \Cref{fig:uhml-sat}. 
	$$ \R\;\defeq\;\setdef{(\eI{\eSem{\hV}}{\pV},\hV)}{\hV{\in}\Sat} $$
 \\[-10mm]
\setcounter{equation}{0}
\begin{proof} We prove this claim by case analysis on the structure of $\hV$.
	
	\begin{case}[\hV=\hVarX]Does not apply since $\hVarX$ is an open formula and thus $\hVarX\notin\Sat$. \end{case}
		
	\begin{case}[\hV=\hFls]Does not apply since $\hFls\notin\Sat$.\end{case}
	
	\begin{case}[\hV=\hTru]Holds trivially since \emph{every process} satisfies \hTru, which thus confirms that $(\eI{\eSem{\hTru}}{\pV},\hTru)\in\R$ according to the definition of \R. \end{case} 
	
	\begin{case}[\hV=\hMaxXF \text{ and } \hVarX{\in}\fv{\hV}] We assume that
		\begin{gather}
			\hMaxXF\in\Sat \label{proof:str-soundness-max-1}
		\end{gather}
		To prove that \R is a satisfaction relation, we show that if $(\eI{\eSem{\hMaxXF}}{\pV},\hMaxXF){\in}\R$, then from the recursive unfolding $\hVMaxXFSub$, we can also synthesise an enforcer $\eSem{\hVMaxXFSub}$ such that $ (\eI{\eSem{\hVMaxXFSub}}{\pV},\hVMaxXFSub)\in\R$ as well. Hence, by \eqref{proof:str-soundness-max-1} and the definition of \Sat we know that $\exists\pV'\cdot\pV'{\hSatS}\hMaxXF$, and so by the definition of \hSatS we can deduce that  $\exists\pV'\cdot\pV'{\hSatS}\hVMaxXFSub$, from which we can thus conclude
		\begin{gather}
			\hVMaxXFSub\in\Sat \label{proof:str-soundness-max-2}
		\end{gather}
		Finally, from \eqref{proof:str-soundness-max-2} and the definition of \R we conclude that
		\begin{gather*}
			(\eI{\eSem{\hVMaxXFSub}}{\pV},\hVMaxXFSub)\in\R 
		\end{gather*}
		as required, and so we are done.
	\end{case}
		
{	\newcommand{\formulaVar}[3]{\displaystyle\hBigAndD{#1{\,\in\,}#2\!\!\!\!\!\!\!\!}\ensuremath{\hNec{\actSN{\pate_{#1}}{\bV_{#1}}}{#3}}}	
	\newcommand{\formulaJ}{\ensuremath{\formulaVar{j}{\IndSet'}{\hFls}}}
	\newcommand{\formulaK}{\ensuremath{\formulaVar{k}{\IndSet''}{\hV_k}}}
	\newcommand{\formula}{\ensuremath{\formulaVar{h}{\IndSet}{\hV_h}}}		
	\newcommand{\formulaSplit}{\ensuremath{\hAnd{\formulaJ\;}{\;\formulaK}}}
	\newcommand{\branchID}{\prf{\actSID{\pate_i}{\bV_i}}{\eSem{\hV_i}}}		
	\newcommand{\branchSUP}{\prf{\actSTD{\pate_i}{\bV_i}}{\mV}}		
	
	\begin{case}[\hV=\formula \text{ and } \bigdistinct{h\in\IndSet}\actSN{\pate_h}{\bV_h}] In this case we will be segmenting the set of indices \IndSet into $\IndSet'$ and $\IndSet''$ such that $\IndSet'$ contains the indices (if any) of the branches where the continuation formula $\hV_i$ is a falsehood \hFls, while $\IndSet''$ contains the rest, and so we will be writing $\formulaSplit$ in lieu of $\formula$.
	We thus assume that 
	\begin{gather}
		\formulaSplit\in\Sat \label{proof:str-soundness-nec-1}
	\end{gather}
	From \eqref{proof:str-soundness-nec-1} and the definition of $\eSem{-}$ we have that
	\begin{gather}
		\eSem{\formulaSplit} = \rec{\rVV}{\Big(\ch{\chBig{j\in\IndSet'}\prf{\actSTD{\pate_j}{\bV_j}}{\rVV}}{\chBig{k\in\IndSet''}\prf{\actSID{\pate_k}{\bV_k}}{\eSem{\hV_k}} }\Big)}=\mV \label{proof:str-soundness-nec-2}
	\end{gather}
	By unfolding the recursive construct in \eqref{proof:str-soundness-nec-2} we have that
	\begin{gather}
		\eSem{\formulaSplit} =  \Big(\ch{\chBig{j\in\IndSet'}\prf{\actSTD{\pate_j}{\bV_j}}{\mV}}{\chBig{k\in\IndSet''}\prf{\actSID{\pate_k}{\bV_k}}{\eSem{\hV_k}} }\Big)
		 \label{proof:str-soundness-nec-2.5}
	\end{gather}
	In order to prove that \R is a satisfaction relation, for this case we must show that every individual branch in \eqref{proof:str-soundness-nec-2.5} is in \R as well. In order to show this we proceed by case analysis and show that the different types of branches that are synthesisable are also in \R. Hence, for all $i\in\IndSet$, we consider the following cases: \medskip	
	\begin{enumerate}[(i)]
		\item when {$\eSem{\hNec{\actSN{\pate_i}{\bV_i}}{\hFls}}{\;=\;}\branchSUP$}:
		In order to prove that this branch is in \R it suffices showing that for all \acta and \pVV, when $\eI{\prf{\actSTD{\pate_i}{\bV_i}}{\mV}}{\pV}\wtraS{\acta}\pVV$ such that $\mtchS{\actSN{\pate_i}{\bV_i}}{\acta}{\,=\,}\sV$ then $(\pVV,\hFls){\,\in\,}\R$. \smallskip
		
		This case holds trivially since by rules \rtit{iTrn} and \rtit{eTrn} we know that whenever \pV produces an action \acta such that symbolic action \actSN{\pate_i}{\bV_i} is satisfied, \ie $\mtchS{\actSN{\pate_i}{\bV_i}}{\acta}{\,=\,}\sV$,  the produced action \acta gets internally transformed into a silent (\actt) action, meaning that $\eI{\branchSUP}{\pV}\nwtraS{\acta}$, and so the modal necessities leading to a falsehood (\eg in this case \hNec{\actSN{\pate_i}{\bV_i}}{\hFls}) \emph{never get satisfied} by the monitored system. \medskip
		\item when {$\eSem{\hNec{\actSN{\pate_i}{\bV_i}}{\hV_i}}{\;=\;}\branchID$}: Once again in order to prove that this branch is in \R, we must show that for all \acta and \pVV, when $\eI{\branchID}{\pV}\wtraS{\acta}\pVV$ such that $\mtchS{\actSN{\pate_i}{\bV_i}}{\acta}{\,=\,}\sV$ then $(\pVV,\hV_i){\,\in\,}\R$. \smallskip
		
		In order to show this we assume that
		\begin{gather}
			\mtchS{\actSN{\pate_i}{\bV_i}}{\acta}=\sV \label{proof:str-soundness-nec-3} \\
			\eI{\branchID}{\pV}\wtraS{\acta}\pVV \label{proof:str-soundness-nec-4}
		\end{gather}
		By the definition of $\wtraS{\acta}$ we know that the weak transition in \eqref{proof:str-soundness-nec-4} is composed from 0 or more \actt-transitions followed by the \acta-transition as shown below
		\begin{gather}
			\eI{\branchID}{\pV}\wtraS{\actt}\pVV'\traS{\acta}\pVV \label{proof:str-soundness-nec-4.1}
		\end{gather}
		By the rules in our model we can infer that the \actt-transitions performed in \eqref{proof:str-soundness-nec-4.1} (if any) are only possible via multiple applications of rule \rtit{iAsy} which allows us to deduce
		\begin{gather}
			\pV\wtraS{\actt}\pV''	\label{proof:str-soundness-nec-4.2}\\
			(\pVV'=\eI{\branchID}{\pV''})\traS{\acta}\pVV  \label{proof:str-soundness-nec-4.3}
		\end{gather}
		Since we do not make any assumptions about the resultant enforced system \pVV, we must first infer its form so to be able to deduce whether $(\pVV,\hV_i\sV)\in\R$ or not. Since the reduction in \eqref{proof:str-soundness-nec-4.3} can be the result of two instrumentation rules, namely \rtit{iTer} and \rtit{iTrn}, we consider both cases separately.
		\begin{itemize}
		 	\item \lipicsHeader{\rtit{iTer}:} As we assume that \eqref{proof:str-soundness-nec-4.3} is the result of rule \rtit{iTer}, by this rule we thus have that $\eI{\branchID}{\pV''}\ntraS{\acta}$, which means that $\mtchS{\actSN{\pate_i}{\bV_i}}{\acta}=\sVundef$ which contradicts with assumption \eqref{proof:str-soundness-nec-3}, and hence this case does not apply.
		 	\item \lipicsHeader{\rtit{iTrn}:} By assuming that \eqref{proof:str-soundness-nec-4.3} is the result of rule \rtit{iTrn}, we thus know that 
		 	\begin{gather}
		 		\pV''\traS{\acta}\pV' \label{proof:str-soundness-nec-5} \\
		 		\pVV = \eI{\eSem{\hV_i\sV}}{\pV'} \label{proof:str-soundness-nec-6} \\
		 		\branchID\traS{\ioact{\acta}{\acta}}\eV' \label{proof:str-soundness-nec-7}
		 	\end{gather}
		 	Hence, from \eqref{proof:str-soundness-nec-6} we know that to prove this case we must show that $(\eI{\eSem{\hV_i\sV}}{\pV},\hV_i\sV){\,\in\,}\R$. We thus refer to our initial assumption \eqref{proof:str-soundness-nec-1} from which by the definition of \Sat we know that there exists some process \pVVV such that $\pVVV\hSatS\formula$. With the definition of \hSatS we thus know that 
		 	\begin{gather}
		 		\exists\pVVV{\,\in\,}\Sys,\forall h{\,\in\,}\IndSet,\pVVV'{\,\in\,}\Sys\cdot \textsl{if } \pVVV\wtraS{\acta}\pVVV' \textsl{ and } \mtchS{\actSN{\pate_h}{\bV_h}}{\acta}{=}\sV \textsl{ then } \pVVV'\hSatS\hV_h\sV \label{proof:str-soundness-nec-9}
		 	\end{gather} 
		 	Since from \eqref{proof:str-soundness-nec-4.2} and \eqref{proof:str-soundness-nec-5} we know that $\pV\wtraS{\acta}\pV'$, and so with the knowledge of \eqref{proof:str-soundness-nec-3}, from \eqref{proof:str-soundness-nec-9} we can thus infer that $\pV'\hSatS\hV_i\sV$ meaning that $\hV_i\sV\in\Sat$. This result allows us to deduce that by the definition of \R we conclude that 
		 	\begin{gather}
		 		(\eI{\eSem{\hV_i\sV}}{\pV'},\hV_i\sV)\in\R \label{proof:str-soundness-nec-10}
		 	\end{gather}
		 	as required. Hence, from assumptions \eqref{proof:str-soundness-nec-3}, \eqref{proof:str-soundness-nec-4} and deduction \eqref{proof:str-soundness-nec-10} we can infer that for $j\in\IndSet$ we know that
		 	\begin{gather*}
		 		(\eI{\branchID}{\pV},\hNec{\actSN{\pate_i}{\bV_i}}{\hV_i})\in\R
		 	\end{gather*} 
		 	as required, and we are done.
\end{itemize}
	\end{enumerate}
	\end{case} 	
}\vspace{-7mm}	
\end{proof} \smallskip
\subsection{Proving Transparency} \label{sec:proof-transparency}
\begin{rtp}
	$$\forall\pV{\,\in\,}\Sys,\hV{\,\in\,}\SHMLnf\cdot \pV{\,\hSatS\,}\hV \; \imp \; \pV\bisim\eI{\eSem{\hV}}{\pV} $$
	To prove this lemma we show that relation \R (below) is a \emph{strong bisimulation relation}.
	$$ \R\;\defeq\;\setdef{(\pV,\eI{\eSem{\hV}}{\pV})}{\pV{\,\hSatS\,}\hV} $$

\noindent Hence we must show that \R satisfies the following transfer properties for each $(\pV,\eI{\eSem{\hV}}{\pV}){\,\in\,}\R$:
	\begin{enumerate}[\quad(a)]
		\item if $\pV\traS{\actu}\pV'$ then $\eI{\eSem{\hV}}{\pV}\traS{\actu}S'$ and $(\pV',S')\in\R$
		\item if $\eI{\eSem{\hV}}{\pV}\traS{\actu}S'$ then $\pV\traS{\actu}\pV'$ and $(\pV',S')\in\R$
	\end{enumerate} 

	\noindent We prove $(a)$ and $(b)$ separately by assuming that $\pV{\,\hSatS\,}\hV$ in both cases as defined by relation \R and conduct these proofs under the assumption that all our formulas are \emph{guarded}, \ie every occurrence of a logical variable \hVarX is always preceded by a modal necessity. It is common knowledge that every \mucalc formula (a reformulation of \recHML) can converted into a semantically equivalent guarded formula of the same logic (see \cite{Banieqbal1898,Walukiewicz2000}). This allows us to conduct the proofs for both $(a)$ and $(b)$ by mathematical induction on the number of maximal fixed points declarations that occur at the \emph{topmost-level} as defined by the rules in \Cref{fig:top-level-max}. 
$\\[-5mm]\;$	
\end{rtp}

\begin{figure}
	$$
	\begin{array}{rcl}
		\lenMax{\hTru} \; = \; \lenMax{\hFls} \; = \; \lenMax{\hVarX} \; = \; \lenMax{\hBigAndD{i\in\IndSet}\hNec{\actS_i}\hV_i } &=& 0 \\
		\lenMax{\hMaxXF} &=& \lenMax{\hV}+1
	\end{array}
	$$
	\caption{The number of top level maximal fixed points.}
	\label{fig:top-level-max}
\end{figure}

\setcounter{equation}{0}
\begin{proof}[Proof for (a)] We proceed by mathematical induction of $\lenMax{\hV}$. 
	
	\begin{Cases}[\lenMax{\hFls}=\lenMax{\hVarX}=0]
		Both cases do not apply since $\nexists \pV\cdot \pV\hSatS\hFls$ and similarly since $\hVarX$ is an open-formula and so $\nexists \pV\cdot \pV\hSatS\hVarX$.
	\end{Cases}

	\begin{case}[\lenMax{\hTru}=0]
		We now assume that 
		\begin{gather}
			\pV\hSatS\hTru \label{proof:trans-tt-1}\\
			\pV\traS{\actu}\pV'  \label{proof:trans-tt-2}
		\end{gather}
		Since $\actu\in\set{\actt,\acta}$, we must consider both cases.
		\begin{itemize}
			\item \lipicsHeader{\actu=\actt:} Since \actu=\actt, we can apply rule \rtit{iAsy} on \eqref{proof:trans-tt-2} and get
				\begin{gather}
					\eI{\eSem{\hTru}}{\pV}\traS{\actt}\eI{\eSem{\hTru}}{\pV'} \label{proof:trans-tt-3}
				\end{gather}
				as required. Also, since we know that every process satisfies \hTru, we know that $\pV'\hSatS\hTru$, which by the definition of \R we conclude
				\begin{gather}
					(\pV',\eI{\eSem{\hTru}}{\pV'})\in\R \label{proof:trans-tt-4}
				\end{gather}
				as required. This means that this subcase is done by \eqref{proof:trans-tt-3} and \eqref{proof:trans-tt-4}.
			\item \lipicsHeader{\actu=\acta:} Since by rule \rtit{eId} we know that $\eIden\traS{\ioact{\acta}{\acta}}\eIden$, and since \actu=\acta, we can apply rule \rtit{iTrn} on \eqref{proof:trans-tt-2} and deduce
				\begin{gather}
					\eI{\eIden}{\pV}\traS{\acta}\eI{\eIden}{\pV'} \label{proof:trans-tt-5}
				\end{gather}
				Since \eSem{\hTru}=\eIden, we can refine \eqref{proof:trans-tt-5} as
				\begin{gather}
					\eI{\eSem{\hTru}}{\pV}\traS{\acta}\eI{\eSem{\hTru}}{\pV'} \label{proof:trans-tt-6}
				\end{gather}
				as required. Once again, since $\pV'\hSatS\hTru$, we can deduce
				\begin{gather}
					(\pV',\eI{\eSem{\hTru}}{\pV'})\in\R \label{proof:trans-tt-7}
				\end{gather}
				as required. This subcase is done by \eqref{proof:trans-tt-6} and \eqref{proof:trans-tt-7}.
		\end{itemize}
	\end{case}

	{\newcommand{\formula}{\hBigAnd{i\in\IndSet}\hNec{\actSN{\pate_i}{\bV_i}}\hV_i}
	\newcommand{\synForm}{\Big(\rec{\rVV}{\chBigI\begin{xbrace}{lc} \prf{\actSTD{\pate_i}{\bV_i}}{\rVV} &\quad (\text{if }\hV_i=\hFls) \\ \prf{\actSID{\pate_i}{\bV_i}}{\eSem{\hV_i}} &\quad (\text{otherwise}) \end{xbrace}}\Big)}
	\begin{case}[\lenMax{\formula}=0]
		Assume that 
		\begin{gather}
			\pV\hSatS\formula \label{proof:trans-nec-1}\\
			\pV\traS{\actu}\pV'  \label{proof:trans-nec-2}
		\end{gather}
		Since $\actu\in\set{\actt,\acta}$, we must consider both cases.
		\begin{itemize}
			\item \lipicsHeader{\actu=\actt:} Since \actu=\actt, we can apply rule \rtit{iAsy} on \eqref{proof:trans-nec-2} and obtain
				\begin{gather}
				\eI{\eSem{\formula}}{\pV}\traS{\actt}\eI{\eSem{\formula}}{\pV'} \label{proof:trans-nec-3}
				\end{gather}
				as required. Since \actu=\actt, and since we know that \SHML is \actt-closed (see Proposition 3.8 in \cite{Aceto1999TestingHML}), from \eqref{proof:trans-nec-1} and \eqref{proof:trans-nec-2}, we can deduce that $\pV'\hSatS\formula$, so that by the definition of \R we conclude 
				\begin{gather}
					(\pV',\eI{\eSem{\formula}}{\pV'})\in\R \label{proof:trans-nec-4}
				\end{gather}
				as required. This subcase is therefore done by \eqref{proof:trans-nec-3} and \eqref{proof:trans-nec-4}.
			\item \lipicsHeader{\actu=\acta:} Since $\actu=\acta$, from \eqref{proof:trans-nec-2} we know that
				\begin{gather}
					\pV\traS{\acta}\pV' \label{proof:trans-nec-5}
				\end{gather}
				Since the branches in our conjunction are all prefixed by disjoint symbolic actions, \ie $\bigdistinct{i\in\IndSet}\actSN{\pate_i}{\bV_i}$, we know that \emph{at most one} of the branches can match an action \acta. Hence, we consider two cases, namely:
				\begin{itemize}
					\item \lipicsHeader{No matching branches (\ie $ \forall i\in\IndSet\cdot\mtch{\actSN{\pate_i}{\bV_i}}{\acta}=\sVundef$):} Since $\eSem{\formula}=\synForm$, and since none of the guarding symbolic transformations in the synthesised selection can match action \acta, we conclude that 
					\begin{gather}
						\eSem{\formula}\ntraS{\acta} \label{proof:trans-nec-6}
					\end{gather}
					Since $\eSem{\hTru}=\eIden$, by \eqref{proof:trans-nec-6} and rule \rtit{iTer} we thus know 
					\begin{gather}
						\eI{\eSem{\formula}}{\pV}\traS{\acta}\eI{\eSem{\hTru}}{\pV'} \label{proof:trans-nec-7}
					\end{gather}
					as required. Also, since any process satisfies \hTru, we know that $\pV'\hSatS\hTru$, and so by the definition of \R we conclude that
					\begin{gather}
						(\pV',\eI{\eSem{\hTru}}{\pV'})\in\R \label{proof:trans-nec-8}
					\end{gather}
					as required. This subcase is therefore done by \eqref{proof:trans-nec-7} and \eqref{proof:trans-nec-8}. \medskip
					\item \lipicsHeader{One matching branch (\ie $ \exists j\in\IndSet\cdot\mtch{\actSN{\pate_j}{\bV_j}}{\acta}=\sV$):} From \eqref{proof:trans-nec-1} and by the definition of \hSatS we know that for every index $i\in\IndSet$ and process $\pV''\in\Sys$ $(\pV\wtraS{\acta}\pV'' \text{ and } \mtchS{\actSN{\pate_i}{\bV_i}}{\acta}{=}\sV)$ \textsl{imply} $\pV''{\hSatS}\hV_j\sV $, and so, since $\exists j\in\IndSet\cdot\mtchS{\actSN{\pate_j}{\bV_j}}{\acta}{=}\sV$, and from  \eqref{proof:trans-nec-5} we can deduce that
					\begin{gather}
						\pV'\hSatS\hV_j\sV \label{proof:trans-nec-9}
					\end{gather}
					Also, since $\mtchS{\actSN{\pate_j}{\bV_j}}{\acta}{=}\sV$, by rule \rtit{eTrn} we know that
					\begin{gather}
						\forall\eV_j,\pate'\cdot\eTrns{\pate_j}{\bV_j}{\pate'}{\eV_j} \traS{\ioact{\acta}{\pate'\sV}} \eV_j\sV \label{proof:trans-nec-10}
					\end{gather}
					By applying rules \rtit{eSel}, \rtit{eRec} on \eqref{proof:trans-nec-10} and then \eqref{proof:trans-nec-2} and \rtit{iTrn} we get
					\begin{gather}
						\forall\eV_j\cdot \eI{\Big((\rec{\rVV}{\ch{(\chBig{k\in\IndSet\setminus\set{j}\hspace{-3mm} }\eTrns{\pate_k}{\bV_k}{\pate'_k}{\eV_k})}{(\eTrns{\pate_j}{\bV_j}{\pate'}{\eV_j})\Big} } )}{\pV} 
						\traS{\pate'\sV} \eI{\eV_j\sV}{\pV'} \label{proof:trans-nec-11}
					\end{gather}
					From \eqref{proof:trans-nec-11} and the definition of \eSem{-} we can infer that $\eV_j=\rVV$ and $\pate'=\actt$ when $\eV_j$ is derived from $\hV_j=\hFls$, or $\eV_j=\eSem{\hV_j}$ and $\pate'=\bnd{\pate_j}$ otherwise. By \eqref{proof:trans-nec-9} we can deduce that the former is \emph{false} because if $\hV_j=\hFls$, then this would contradict with \eqref{proof:trans-nec-9}, and hence only the latter applies. So, since $\eSem{\hV_j\sV}=\eV_j\sV$ and $\pate_j\sV=\acta$ we have that 
					\begin{gather}
						\forall\eV_j\cdot \eI{\Big((\rec{\rVV}{\ch{(\chBig{k\in\IndSet\setminus\set{j}\hspace{-3mm} }\eTrns{\pate_k}{\bV_k}{\pate'_k}{\eV_k})}{(\eTrns{\pate_j}{\bV_j}{\bnd{\pate_j}}{\eV_j})\Big} } )}{\pV} 
						\traS{\acta} \eI{\eSem{\hV_j\sV}}{\pV'} \label{proof:trans-nec-12}
					\end{gather}
					By \eqref{proof:trans-nec-12} and the definition of \eSem{-} we can thus conclude that
					\begin{gather}
						\eI{\eSem{\formula}}{\pV}\traS{\acta}\eI{\eSem{\hV_j\sV}}{\pV'} \label{proof:trans-nec-13}
					\end{gather}
					as required, and by \eqref{proof:trans-nec-9} and the definition of \R we conclude that 
					\begin{gather}
						(\pV',\eI{\eSem{\hV_j\sV}}{\pV'})\in\R \label{proof:trans-nec-14}
					\end{gather}
					as required. Hence, this subcase is done by \eqref{proof:trans-nec-13} and \eqref{proof:trans-nec-14}.
				\end{itemize}
		\end{itemize}
	\end{case}
	}
	
	{\newcommand{\formula}{\hMaxXF}
	\newcommand{\synForm}{\rec{\rV}{\eSem{\hV}}}
		\begin{case}[\lenMax{\formula}=l+1]
			We start by assuming that 
			\begin{gather}
			\pV\hSatS\formula \label{proof:trans-max-1}\\
			\pV\traS{\actu}\pV'  \label{proof:trans-max-2}
			\end{gather}
			Since $\actu\in\set{\actt,\acta}$, we must consider both cases.
			\begin{itemize}
				\item \lipicsHeader{\actu=\actt:} Since \actu=\actt, we can apply rule \rtit{iAsy} on \eqref{proof:trans-max-2} and deduce that
				\begin{gather}
					\eI{\eSem{\formula}}{\pV}\traS{\actt}\eI{\eSem{\formula}}{\pV'} \label{proof:trans-maxa-3}
				\end{gather}
				as required. Also, since \SHML is \actt-closed (see Proposition 3.8 in \cite{Aceto1999TestingHML}), by \eqref{proof:trans-max-1} and \eqref{proof:trans-max-2}, we also know that $\pV'\hSatS\formula$ as well. Hence, by the definition of \R we conclude 
				\begin{gather}
					(\pV',\eI{\eSem{\formula}}{\pV'})\in\R \label{proof:trans-maxa-4}
				\end{gather}
				and so we done by \eqref{proof:trans-maxa-3} and \eqref{proof:trans-maxa-4}.
				\item \lipicsHeader{\actu=\acta:} Since $\actu=\acta$, from \eqref{proof:trans-max-2} we know that
				\begin{gather}
					\pV\traS{\acta}\pV' \label{proof:trans-maxb-3}
				\end{gather}
				and by \eqref{proof:trans-max-1} and the definition of \hSatS we know
				\begin{gather}
					\pV\hSatS\hVMaxXFSub \label{proof:trans-maxb-4}
				\end{gather}
				Since we assume that logical variables (\eg \hVarX) are \emph{guarded}, by the definition of \lenMax{\hV} we know that whenever a maximal fixed point \hMaxXF is unfolded into \hVMaxXFSub, the number of top level maximal fixed points decreases by 1, and so since $\lenMax{\hMaxXF}=l+1$, we infer that 
				\begin{gather}
					\lenMax{\hVMaxXFSub}=l  \label{proof:trans-maxb-5}
				\end{gather}
				Hence, by \eqref{proof:trans-maxb-3}, \eqref{proof:trans-maxb-4}, \eqref{proof:trans-maxb-5} and the inductive hypothesis we can deduce that
				\begin{gather}
					\exists\pVV'\cdot \eI{\eSem{\hVMaxXFSub}}{\pV} \traS{\acta} \pVV' \label{proof:trans-maxb-6} \\
					(\pV',\pVV')\in\R  \label{proof:trans-maxb-7}
				\end{gather}
				By applying the definition of \eSem{-} on \eqref{proof:trans-maxb-6}, followed by rule \rtit{iTrn} we get
				\begin{gather}
					\exists\pVV'\cdot \eSem{\hV}\Sub{\synForm}{\rV} \traS{\ioact{\acta}{\acta}}\eV \qquad \text{ where }\pVV'=\eI{\eV}{\pV'} \label{proof:trans-maxb-8}
				\end{gather}
				By applying rule \rtit{eRec} on \eqref{proof:trans-maxb-8}, followed by  \eqref{proof:trans-maxb-3} and \rtit{iTrn} we get
				\begin{gather}
					\exists\pVV'\cdot \eI{\synForm}{\pV} \traS{\acta}\pVV' \label{proof:trans-maxb-9}
				\end{gather}
				and so, we can apply \eSem{-} on \eqref{proof:trans-maxb-9} and obtain
				\begin{gather}
					\exists\pVV'\cdot \eI{\eSem{\formula}}{\pV} \traS{\acta}\pVV' \label{proof:trans-maxb-10}
				\end{gather}
				as required. We are therefore done by \eqref{proof:trans-maxb-7} and \eqref{proof:trans-maxb-10}.
			\end{itemize}
		\end{case}
	}		
\end{proof}

\begin{proof}[Proof for (b)] The proof proceeds by mathematical induction of $\lenMax{\hV}$. 
	
	\begin{Cases}[\lenMax{\hFls}=\lenMax{\hVarX}=0]
		Both cases do not apply since $\nexists \pV\cdot \pV\hSatS\hFls$ and similarly since $\hVarX$ is an open-formula and $\nexists \pV\cdot \pV\hSatS\hVarX$.
	\end{Cases}
		
	\begin{case}[\lenMax{\hTru}=0]
		Assume that 
		\begin{gather}
			\pV\hSatS\hTru \label{proof:trans-b-tt-1}\\
			\eI{\eSem{\hTru}}{\pV}\traS{\actu}\pVV'  \label{proof:trans-b-tt-2}
		\end{gather}
		Since $\actu\in\set{\actt,\acta}$, we must consider both cases.
		\begin{itemize}
			\item \lipicsHeader{\actu=\actt:} Since \actu=\actt, the transition in \eqref{proof:trans-b-tt-2} can be performed either via \rtit{iTrn} or \rtit{iAsy}. We must therefore consider both cases.
			\begin{itemize}
				\item \lipicsHeader{\rtit{iAsy}:} From rule \rtit{iAsy} and \eqref{proof:trans-b-tt-2} we thus know that $\pVV'=\eI{\eV}{\pV'}$ and that $\eV=\eSem{\hTru}$ since this remains unaffected by the transition, such that $\pV\traS{\actt}\pV'$ as required. Also, since every process satisfies \hTru, we know that $\pV'\hSatS\hTru$ as well, and so we are done since by the definition of \R we know that $(\pV',\eI{\eSem{\hTru}}{\pV'})\in\R$.
				\item \lipicsHeader{\rtit{iTrn}:} From rule \rtit{iTrn} and \eqref{proof:trans-b-tt-2} we know that: $\pVV'=\eI{\eV}{\pV'}$, $\pV\traS{\acta}\pV'$ and that
				\begin{gather}
					\eSem{\hTru}\traS{\ioact{\acta}{\actt}}\eV \label{proof:trans-b-tt-3}
				\end{gather}
				Since $\eSem{\hTru}=\eIden$, by rule \rtit{eId} we know that \eqref{proof:trans-b-tt-3} is \emph{false} and hence this case does not apply.
			\end{itemize}
			\item \lipicsHeader{\actu=\acta:} Since \actu=\acta, the transition in \eqref{proof:trans-b-tt-2} can be performed either via \rtit{iTrn} or \rtit{iTer}. We consider both cases.
			\begin{itemize}
				\item \lipicsHeader{\rtit{iTer}:} This case does not apply since by applying \rtit{iTer} on  \eqref{proof:trans-b-tt-2} we know that $\eSem{\hTru}\ntraS{\acta}$ which is \emph{false} since $\eSem{\hTru}=\eIden$ and rule \rtit{eId} states that for all \acta, $\eIden\traS{\ioact{\acta}{\acta}}\eIden$, thus leading to a contradiction.
				\item \lipicsHeader{\rtit{iTrn}:} By applying rule \rtit{iTrn} on \eqref{proof:trans-b-tt-2} we know that $\pVV'=\eI{\eV}{\pV'}$ such that
				\begin{gather}
					\pV\traS{\acta}\pV' \label{proof:trans-b-tt-4} \\
					\eSem{\hTru}\traS{\ioact{\acta}{\acta}}\eV \label{proof:trans-b-tt-5}
				\end{gather}
				Since $\eSem{\hTru}=\eIden$, by applying rule \rtit{eId} to \eqref{proof:trans-b-tt-5} we know that $\eV=\eIden=\eSem{\hTru}$, meaning that $\pVV'=\eI{\eSem{\hTru}}{\pVV'}$. Hence, since every process satisfies \hTru we know that $\pV'\hSatS\hTru$, so that by the definition of \R we conclude 
				\begin{gather}
					(\pV',\eI{\eSem{\hTru}}{\pV'})\in\R \label{proof:trans-b-tt-6}
				\end{gather}
				Hence, we are done by \eqref{proof:trans-b-tt-4} and \eqref{proof:trans-b-tt-6}.
			\end{itemize}
		\end{itemize}
	\end{case}
	
	{\newcommand{\formula}{\hBigAnd{i\in\IndSet}\hNec{\actSN{\pate_i}{\bV_i}}\hV_i}
		\newcommand{\synForm}{\Big(\rec{\rVV}{\chBigI\begin{xbrace}{lc} \prf{\actSTD{\pate_i}{\bV_i}}{\rVV} &\quad (\text{if }\hV_i=\hFls) \\ \prf{\actSID{\pate_i}{\bV_i}}{\eSem{\hV_i}} &\quad (\text{otherwise}) \end{xbrace}}\Big)}
		\newcommand{\synFormUnfold}{\Big(\chBigI\begin{xbrace}{lc} \prf{\actSTD{\pate_i}{\bV_i}}{\rVV} &\quad (\text{if }\hV_i=\hFls) \\ \prf{\actSID{\pate_i}{\bV_i}}{\eSem{\hV_i}} &\quad (\text{otherwise}) \end{xbrace}\Big)\sub{\eV}{\rVV}}
		\begin{case}[\lenMax{\formula}=0]
			We assume that 
			\begin{gather}
			\pV\hSatS\formula \label{proof:trans-nec-b-1}\\
			\eI{\eSem{\formula}}{\pV}\traS{\actu}\pVV'  \label{proof:trans-nec-b-2}
			\end{gather}
			Since $\actu\in\set{\actt,\acta}$, we must consider both cases.
			\begin{itemize}
				\item \lipicsHeader{\actu=\actt:} Since \actu=\actt, from \eqref{proof:trans-nec-b-2} we know that
				\begin{gather}
					\eI{\eSem{\formula}}{\pV}\traS{\actt}\pVV'  \label{proof:trans-nec-b-3}
				\end{gather}
				The \actt-transition in \eqref{proof:trans-nec-b-3} can be performed either via \rtit{iTrn} or \rtit{iAsy}; we thus consider both cases.
				\begin{itemize}
					\item \lipicsHeader{\rtit{iAsy}:} As we assume that the reduction in \eqref{proof:trans-nec-b-3} is the result of rule \rtit{iAsy}, we know that $\pVV'=\eI{\eSem{\formula}}{\pV'}$ such that
					\begin{gather}
						\pV\traS{\actt}\pV'  \label{proof:trans-nec-b-4}
					\end{gather}
					as required. Also, since \SHML is \actt-closed (see Proposition 3.8 in \cite{Aceto1999TestingHML}), by \eqref{proof:trans-nec-b-1} and \eqref{proof:trans-nec-b-4} we can deduce that $\pV'\hSatS\formula$ as well, so that by the definition of \R we conclude that
					\begin{gather}
						(\pV',\eI{\eSem{\formula}}{\pV'})\in\R  \label{proof:trans-nec-b-5}
					\end{gather}
					and so we are done by \eqref{proof:trans-nec-b-4} and \eqref{proof:trans-nec-b-5}.
					\item \lipicsHeader{\rtit{iTrn}:} By assuming that reduction \eqref{proof:trans-nec-b-3} results from \rtit{iTrn}, we know that $\pVV'=\eI{\eV'}{\pV'}$ such that
					\begin{gather}
						\pV\traS{\acta}\pV' \label{proof:trans-nec-b-5.1} \\
						\eSem{\formula}\traS{\ioact{\acta}{\actt}}\eV' \label{proof:trans-nec-b-6}
					\end{gather}
					By \eqref{proof:trans-nec-b-6} and the definition of \eSem{-} we know that
					\begin{gather}
						(\eV=\synForm) \traS{\ioact{\acta}{\actt}} \eV' \label{proof:trans-nec-b-7}
					\end{gather}
					By applying rule \rtit{eRec} on \eqref{proof:trans-nec-b-7} we know
					\begin{gather}
						\synFormUnfold \traS{\ioact{\acta}{\actt}} \eV' \label{proof:trans-nec-b-8}
					\end{gather}
					From \eqref{proof:trans-nec-b-8} we know that the input action \acta is suppressed into a \actt which is only possible when \acta matches a branch of the form \prf{\actSTD{\pate_j}{\bV_j}}{\rVV} for some $j\in\IndSet$, and so we know that
					\begin{gather}
						\exists j\in\IndSet\cdot\mtchS{\actSN{\pate_j}{\bV_j}}{\acta}=\sV \label{proof:trans-nec-b-8.1}
					\end{gather}
					By the definition of \eSem{-} we however know that this matching branch was derived from a conjunct subformula of the form $\hNec{\actSN{\pate_j}{\bV_j}}\hFls$, such that we know that
					\begin{gather}
						\pVV'=\eI{\eSem{\hFls}}{\pV'}	\label{proof:trans-nec-b-9}
					\end{gather}
					According to the definition of \R, for the pair $(\pV',\pVV')$ to be in \R we must now show that $\pV'\hSatS\hFls$ which is obviously \emph{false}, and hence, contradicts with assumption \eqref{proof:trans-nec-b-1}. Precisely, this contradiction occurs since by the definition of \hSatS, when $\pV\hSatS\formula$ then $\pV\wtraS{\acta}\pV'$ (which is confirmed by \eqref{proof:trans-nec-b-5.1}) and $\exists j\in\IndSet\cdot\mtchS{\actSN{\pate_j}{\bV_j}}{\acta}{=}\sV$ (also confirmed by \eqref{proof:trans-nec-b-8.1}) imply $\pV'\hSatS\hV_j\sV$ which leads to a contradiction since in this case $\hV_j\sV{=}\hFls$. Hence, this subcase does not apply.					
				\end{itemize}
			\item \lipicsHeader{\actu=\acta:} Since \actu=\acta, by \eqref{proof:trans-nec-b-2} and the definition of \eSem{-} we know that
			\begin{gather}
				\eI{\synForm}{\pV}\traS{\acta}\pVV'  \label{proof:trans-nec-b-10}
			\end{gather}	
				Since the transition in \eqref{proof:trans-nec-b-10} can be performed via \rtit{iTer} or iTrn, we consider both possibilities.
				\begin{itemize}
					\item \lipicsHeader{\rtit{iTer}:} As we assume that \eqref{proof:trans-nec-b-10} results from rule \rtit{iTer}, we know that 
					\begin{gather}
						\pV\traS{\acta}\pV' \label{proof:trans-nec-b-11}
					\end{gather}
					as required, and that $\pVV'=\eI{\eIden}{\pV'}=\eI{\eSem{\hTru}}{\pV'}$ since $\eSem{\hTru}=\eIden$. Consequently, as every process satisfies \hTru, we know that $\pV'\hSatS\hTru$ and so by the definition of \R we can conclude that
					\begin{gather}
						(\pV',\eI{\eSem{\hTru}}{\pV'})\in\R \label{proof:trans-nec-b-12}
					\end{gather}
					and so we are done by \eqref{proof:trans-nec-b-11} and \eqref{proof:trans-nec-b-12}.
					\item \lipicsHeader{\rtit{iTrn}:} By assuming that \eqref{proof:trans-nec-b-10} is obtained from rule \rtit{iTrn} we know that
					\begin{gather}
						\pV\traS{\acta}\pV' \label{proof:trans-nec-b-13}
					\end{gather}
					as required, and that 
					\begin{gather}
						\synForm\traS{\ioact{\acta}{\acta}}\pVV' \label{proof:trans-nec-b-14}
					\end{gather}
					By applying rules \rtit{eRec} and \rtit{eSel} on \eqref{proof:trans-nec-b-14} we know
					\begin{gather}
						\exists j\in\IndSet\cdot \eTrns{\pate_j}{\bV_j}{\pate'}{\eV} \traS{\ioact{\acta}{\acta}}\pVV' \label{proof:trans-nec-b-15}
					\end{gather}
					Since the transition in \eqref{proof:trans-nec-b-15} does not modify the given action \acta, we can infer that $\pate'=\bnd{\pate_j}$ and that $\eV=\eSem{\hV_j}$ where $\hV_j\neq\hFls$ so that when we apply rule \rtit{eTrn} to \eqref{proof:trans-nec-b-15} we can deduce that 
					\begin{gather}
						\pVV'=\eSem{\hV_j\sV}  \label{proof:trans-nec-b-16}\\
						\mtchS{\actSN{\pate_j}{\bV_j}}{\acta}=\sV  \label{proof:trans-nec-b-17}
					\end{gather}
					By applying the definition of \hSatS on \eqref{proof:trans-nec-b-1} we know that
					\begin{gather}
						\forall i\in\IndSet,\pV''\cdot \pV\wtraS{\acta}\pV'' \text{ and } \mtchS{\actSN{\pate_i}{\bV_i}}{\acta}=\sV \text{ then } \pV''\hSatS\hV_i\sV  \label{proof:trans-nec-b-18}
					\end{gather}
					Hence, from \eqref{proof:trans-nec-b-13}, \eqref{proof:trans-nec-b-17} and \eqref{proof:trans-nec-b-18} we can deduce that $\pV'\hSatS\hV_j\sV$ and so by the definition of \R we can deduce that
					\begin{gather}
						(\pV',\eI{\eSem{\hV_j\sV}}{\pV'})\in\R  \label{proof:trans-nec-b-19}
					\end{gather}
					and so we are done by \eqref{proof:trans-nec-b-13} and \eqref{proof:trans-nec-b-19}.
				\end{itemize}
			\end{itemize}
		\end{case}
	}
	
	{\newcommand{\formula}{\hMaxXF}
		\newcommand{\synForm}{\rec{\rV}{\eSem{\hV}}}
		\begin{case}[\lenMax{\formula}=l+1]
			Assume that 
			\begin{gather}
			\pV\hSatS\formula \label{proof:trans-max-b-1}\\
			\eI{\eSem{\formula}}{\pV}\traS{\actu}\pVV'  \label{proof:trans-max-b-2}
			\end{gather}
			Since the reduction in \eqref{proof:trans-max-b-2} can be performed as a result of rules \rtit{iAsy}, \rtit{iTer} and \rtit{iTrn}, we consider each case separately.
			\begin{itemize}
				\item \lipicsHeader{\rtit{iAsy}:} From rule \rtit{iAsy} and \eqref{proof:trans-max-b-2} we get that $\actu=\actt$ and that
				\begin{gather}
					\pV\traS{\actt}\pV' \label{proof:trans-max-b-4}
				\end{gather}
				as required, and that $\pVV'=\eI{\eSem{\hMaxXF}}{\pV'}$. Hence, since \SHML is \actt-closed (as advocated by Proposition 3.8 in \cite{Aceto1999TestingHML}) by  \eqref{proof:trans-max-b-1} and \eqref{proof:trans-max-b-4} we deduce that $\pV'\hSatS\formula$ as well, and so by the definition of \R we conclude 
				\begin{gather}
					(\pV',\eI{\eSem{\formula}}{\pV'})\in\R  \label{proof:trans-max-b-5}
				\end{gather}
				as required. We are therefore done by \eqref{proof:trans-max-b-4} and \eqref{proof:trans-max-b-5}.				
				\item \lipicsHeader{\rtit{iTer}:} If we assume that \eqref{proof:trans-max-b-2} results from rule \rtit{iTer}, we get that $\actu=\acta$ and that
				\begin{gather}
					\pV\traS{\acta}\pV' \label{proof:trans-max-b-6}
				\end{gather}
				as required, and that $\pVV'=\eI{\eIden}{\pV'}=\eI{\eSem{\hTru}}{\pV'}$, since $\eSem{\hTru}=\eIden$. Hence, since \hTru is \emph{always satisfied}, we know that $\pV'\hSatS\hTru$ and so by the definition of \R we can conclude 
				\begin{gather}
					(\pV',\eI{\eSem{\hTru}}{\pV'})\in\R \label{proof:trans-max-b-7}
				\end{gather}
				Hence, we are done by \eqref{proof:trans-max-b-6} and \eqref{proof:trans-max-b-7}.
				\item \lipicsHeader{\rtit{iTrn}:} By assuming that the reduction in \eqref{proof:trans-max-b-2} was performed via rule \rtit{iTrn} and by the definition of \eSem{-}, we know that
				\begin{gather}
					\pV\traS{\acta}\pV' \label{proof:trans-max-b-8}\\
					\rec{\rV}{\eSem{\hV}}\traS{\ioact{\acta}{\actu}}\eV \qquad (\text{where }\pVV=\eI{\eV}{\pV'}) \label{proof:trans-max-b-9}
				\end{gather}
				By applying rule \rtit{eRec} to \eqref{proof:trans-max-b-9}, along with the definition of, \eSem{-} we can deduce that $\eSem{\hVMaxXFSub}\traS{\ioact{\acta}{\actu}}\eV$, so that by \eqref{proof:trans-max-b-8} and \rtit{iTrn} we have that
				\begin{gather}
					\eI{\eSem{\hVMaxXFSub}}{\pV}\traS{\actu}\pVV'	\label{proof:trans-max-b-10}
				\end{gather}
				By \eqref{proof:trans-max-b-1} and the definition of \hSatS we know that 
				\begin{gather}
					\pV\hSatS\hVMaxXFSub	\label{proof:trans-max-b-11}
				\end{gather}
				Since we assume that logical variables (\eg \hVarX) are \emph{guarded}, by the definition of \lenMax{\hV} we know that whenever a maximal fixed point \hMaxXF gets unfolded into \hVMaxXFSub, the number of top level maximal fixed points decreases by 1, and so, since $\lenMax{\hMaxXF}=l+1$ we infer that
				\begin{gather}
					\lenMax{\hVMaxXFSub}=l \label{proof:trans-max-b-12}
				\end{gather}
				Hence, by \eqref{proof:trans-max-b-10}, \eqref{proof:trans-max-b-11}, \eqref{proof:trans-max-b-12} and the inductive hypothesis we conclude that $\pV\traS{\actu}\pV'$ and $(\pV',\pVV')\in\R$ as required, and so we are done.
			\end{itemize}
		\end{case}
	} \vspace{-5mm}
\end{proof}
\medskip

\subsection{Non-Violating Trace Transparency} \label{sec:proof-trace-transparency}
\begin{rtp}
	\begin{enumerate}[\qquad(a)]
		\item $\forall\pV{\,\in\,}\Sys, \hV{\,\in\,}\SHMLnf \cdot \nvsat{\pV}{\tr}{\hV}  \,\text{ and }\, \pV\wtraS{\tr}\pV'\; \imply \; \eI{\eSem{\hV}{}}{\pV}\wtraS{\tr}\eI{\eV'}{\pV'}$
		\item $\forall\pV{\,\in\,}\Sys, \hV{\,\in\,}\SHMLnf \cdot \nvsat{\pV}{\tr}{\hV}  \,\text{ and }\, \eI{\eSem{\hV}{}}{\pV}\wtraS{\tr}\eI{\eV'}{\pV'} \; \imply \;  \pV\wtraS{\tr}\pV'$
	\end{enumerate}\medskip
	The proofs for (a) and (b) rely on a number of auxiliary lemmas, namely, \Cref{lemma:nvtt-1,lemma:nvtt-2} are required for proving (a) while \Cref{lemma:nvtt-1,lemma:nvtt-3,,lemma:nvtt-4} are necessary for proving (b). Before introducing these lemmas, in \Cref{fig:after-defs} we introduce function 
	 $\afterFS{::}(\SHMLnf\times\Act)\mapsto\SHMLnf$, denoting how 
	 an \SHMLnf formula 
	 is affected after 
	 evaluating with respect to some action \actu.
	\begin{figure}[t]
		\begin{align*}
		\afterF{\hV}{\actt} &\defeq \hV \\
		\afterF{\hV}{\acta} &\defeq 
		\begin{xbrace}{c@{\qquad}l}
		\hV & \text{if }\hV{\,\in\,}\Set{\hTru,\hFls}\\
		\afterF{\hV\sub{\hMaxXF}{\hVarX}}{\acta} & \text{if }\hV{\,=\,}\hMaxXF\\
		\hV_j\sV & \text{if }\hV{\,=\,}\hBigAnd{i\in\IndSet}\hNec{\actS_i}\hV_i\;\text{ and }\;\exists j\in\IndSet\cdot\mtchS{\actS_j}{\acta}{=}\sV \\
		\hTru &\text{if }\hV{\,=\,}\hBigAnd{i\in\IndSet}\hNec{\actS_i}\hV_i\;\text{ and }\;\emph{otherwise}
		\end{xbrace}
		\end{align*}
		\caption{Defining function \afterFS.}
		\label{fig:after-defs}
	\end{figure}
	
	\begin{lemma} \label[lemma]{lemma:nvtt-1}
		$$\pV\wtraS{\acta}\pV' \text{ and } \nvsat{\pV}{\acta\tr}{\hV} \imp \nvsat{\pV'}{\tr}{\afterF{\hV}{\acta}}$$
		This lemma states that if process \pV does not violate \hV \wrt trace, $\acta\tr$, then the process resulting from performing action \acta, \ie	$\pV'$ and the trace suffix \tr, should also not violate the \SHMLnf formula obtained after \hV analyses action \acta, \ie \afterF{\hV}{\acta}. \qed
	\end{lemma}

	\begin{lemma} \label[lemma]{lemma:nvtt-2}
		$$\nvsat{\pV}{\acta\tr}{\hV} \text{ and } \pV\wtraS{\acta}\pV' \imp \eI{\eSem{\hV}}{\pV}\wtraS{\acta}\eI{\eSem{\afterF{\hV}{\acta}}}{\pV'} $$
		This lemma dictates that if process \pV does not violate \hV \wrt trace $\acta\tr$, \ie $\nvsat{\pV}{\acta\tr}{\hV}$, and is capable of performing \acta, \ie $\pV\wtraS{\acta}\pV'$, 
		then the enforced process $\eI{\eSem{\hV}}{\pV}$ should still be able to perform action \acta and reduce into $\eI{\eSem{\afterF{\hV}{\acta}}}{\pV'}$. \qed
	\end{lemma}

	\begin{lemma} \label[lemma]{lemma:nvtt-3}
		$$\nvsat{\pV}{\tr}{\hV} \text{ and } \eI{\eSem{\hV}}{\pV}\traS{\actt}\eI{\eV'}{\pV'} \imp \pV\traS{\actt}\pV' \text{ and } \eV'=\eSem{\hV} \text{ and } \nvsat{\pV'}{\tr}{\hV} $$
		With this lemma we can deduce that if process \pV does not violate \hV \wrt any trace $\tr$, \ie $\nvsat{\pV}{\tr}{\hV}$, and when instrumented with monitor \eSem{\hV} it is capable of performing a silent \actt action, \ie $\eI{\eSem{\hV}}{\pV}\wtraS{\actt}\eI{\eV'}{\pV'}$, then $\eV'$ should still be equal to \eSem{\hV} and the unmonitored process $\pV$ should also be able to perform the same silent action and reduce into $\pV'$ such that this process also does not violate \hV \wrt the same trace \tr. \qed
	\end{lemma}

	\begin{lemma} \label[lemma]{lemma:nvtt-4}
		$$\eI{\eSem{\hV}}{\pV}\traS{\acta}\eI{\eV'}{\pV'} \imp \pV\traS{\acta}\pV' \text{ and } \eV'=\eSem{\afterF{\hV}{\acta}} $$
		This lemma is similar \Cref{lemma:nvtt-3} but applies for visible actions. \qed
	\end{lemma}

	We first prove our main result, \ie implications (a) and (b) of the Non-Violating Trace Transparency, by assuming that these auxiliary lemmas hold; we then prove them afterwards.
\end{rtp}

\setcounter{equation}{0}
\begin{proof}[Proof for (a)] By induction on the length of trace $\tr$. 
	
	\begin{case}[\tr=\varepsilon]
		We assume that $\nvsat{\pV}{\varepsilon}{\hV}$ and that 
		\begin{gather}
			\pV\wtraS{\varepsilon}\pV'  \label{proof:ntt-bc-1}
		\end{gather}
		From the definition of $\wtraS{\varepsilon}$ and \eqref{proof:ntt-bc-1} we know that $\pV\traS{\actt}^{\!\ast}\pV'$, and hence by zero or more applications of \rtit{iAsy} we infer that $\eI{\eSem{\hV}}{\pV}\traS{\actt}^{\!\ast}\eI{\eSem{\hV}}{\pV'}$ and so by the definition of \wtraS{\tr}, we conclude that 
			$$\eI{\eSem{\hV}}{\pV}\wtraS{\varepsilon}\eI{\eSem{\hV}}{\pV'}$$
		as required.
	\end{case}


	\begin{case}[\forall\trr\cdot \tr=\acta\trr]
		We start by assuming that
		\begin{gather}
			\pV\wtraS{\acta\trr}\pV' \label{proof:ntt-ic-1} \\
			\nvsat{\pV}{\acta\trr}{\hV} \label{proof:ntt-ic-2}
		\end{gather}
		By \eqref{proof:ntt-ic-1} and the definition of \wtraS{\tr}, we have that
		\begin{gather}
			\pV\wtraS{\acta}\pV'' \label{proof:ntt-ic-3}\\
			\pV''\wtraS{\trr}\pV' \label{proof:ntt-ic-4}
		\end{gather}
		and by \eqref{proof:ntt-ic-2}, \eqref{proof:ntt-ic-3} and \Cref{lemma:nvtt-1} we know that
		\begin{gather}
			\nvsat{\pV''}{\trr}{\afterF{\hV}{\acta}} \label{proof:ntt-ic-6}
		\end{gather}
		With the knowledge of \eqref{proof:ntt-ic-4} and \eqref{proof:ntt-ic-6} we can now apply the \emph{inductive hypothesis} and infer that
		\begin{gather}
		\eI{\eSem{\afterF{\hV}{\acta}}}{\pV''}\wtraS{\trr}\eI{\eV'}{\pV'}. \label{proof:ntt-ic-9}
		\end{gather}
		Following this, by \eqref{proof:ntt-ic-2}, \eqref{proof:ntt-ic-3} and \Cref{lemma:nvtt-2} we have that
		\begin{gather}
			\eI{\eSem{\hV}}{\pV}\wtraS{\acta}\eI{\eSem{\afterF{\hV}{\acta}}}{\pV''} \label{proof:ntt-ic-8}
		\end{gather}		
		Finally, by joining together \eqref{proof:ntt-ic-9} and \eqref{proof:ntt-ic-8} with the definition of \wtraS{\tr} we can conclude that
		\begin{gather*}
			\eI{\eSem{\hV}}{\pV}\wtraS{\acta\trr}\eI{\eSem{\afterF{\hV}{\acta}}}{\pV'} 
		\end{gather*}
		as required, and so we are done. \vspace{-5mm}
	\end{case}
\end{proof}\pagebreak

\setcounter{equation}{0}
\begin{proof}[Proof for (b)] By induction on the length of trace $\tr$. 
	
	\begin{case}[\tr=\varepsilon]
		We assume that 
		\begin{gather}
		\nvsat{\pV}{\varepsilon}{\hV} \label{proof:ntt-b-bc-2} \\
		\eI{\eSem{\hV}}{\pV}\wtraS{\varepsilon}\eI{\eV'}{\pV'}   \label{proof:ntt-b-bc-1}
		\end{gather}
		From \eqref{proof:ntt-b-bc-1} and the definition of $\wtraS{\varepsilon}$ we know that
		\begin{gather}
			\eI{\eSem{\hV}}{\pV}\traS{\actt}^{\ast}\eI{\eV'}{\pV'}   \label{proof:ntt-b-bc-3}
		\end{gather}
		We now consider two cases for  \eqref{proof:ntt-b-bc-3}, namely, when $\traS{\actt}^{0}$ and $\traS{\actt}\cdot\wtraS{\varepsilon}$.
		\begin{itemize}
			\item when $\traS{\actt}^{0}$: Since no transitions have been applied, from  \eqref{proof:ntt-b-bc-3} we know that $\eV'=\eSem{\hV}$ and $\pV'=\pV$ and so by the definition of $\wtraS{\varepsilon}$ we can immediately conclude that $\pV\wtraS{\varepsilon}\pV'$ as required.
			\item when $\traS{\actt}\cdot\wtraS{\varepsilon}$: From \eqref{proof:ntt-b-bc-3} we can now deduce that 
			\begin{gather}
				\eI{\eSem{\hV}}{\pV}\traS{\actt}\eI{\eV''}{\pV''}   \label{proof:ntt-b-bc-4} \\
				\eI{\eV''}{\pV''}\wtraS{\varepsilon}\eI{\eV'}{\pV'}   \label{proof:ntt-b-bc-5} 
			\end{gather}
			and so by \eqref{proof:ntt-b-bc-2}, \eqref{proof:ntt-b-bc-4} and \Cref{lemma:nvtt-3} we can infer that 
			\begin{gather}
				\pV\traS{\actt}\pV'' \label{proof:ntt-b-bc-6} \\
				\eV''=\eSem{\hV}  \label{proof:ntt-b-bc-7} \\
				\nvsat{\pV''}{\varepsilon}{\hV} \label{proof:ntt-b-bc-8}
			\end{gather}
			Hence, by  \eqref{proof:ntt-b-bc-5}, \eqref{proof:ntt-b-bc-7}, \eqref{proof:ntt-b-bc-8} and the inductive hypothesis we conclude that
			\begin{gather}
				\pV''\wtraS{\varepsilon}\pV' \label{proof:ntt-b-bc-9}
			\end{gather}
			and so we can conclude by \eqref{proof:ntt-b-bc-6} and \eqref{proof:ntt-b-bc-9} that 
			\begin{gather*}
				\pV\wtraS{\varepsilon}\pV'
			\end{gather*}
			as required.
		\end{itemize}
	\end{case}
	
	\begin{case}[\forall\trr\cdot \tr=\acta\trr]
		We first assume that
		\begin{gather}
			\eI{\eSem{\hV}}{\pV}\wtraS{\acta\trr}\eI{\eV'}{\pV'} \label{proof:ntt-b-ic-1} \\
			\forall\trr\cdot\nvsat{\pV}{\acta\trr}{\hV} \label{proof:ntt-b-ic-2}
		\end{gather}
		By \eqref{proof:ntt-b-ic-1} and the definition of \wtraS{\tr}, we have that
		\begin{gather}
			\eI{\eSem{\hV}}{\pV}\wtraS{\acta}\eI{\eV''}{\pV''} \label{proof:ntt-b-ic-3}\\
			\eI{\eV''}{\pV''}\wtraS{\trr}\eI{\eV'}{\pV'} \label{proof:ntt-b-ic-4}
		\end{gather}
		and by \eqref{proof:ntt-b-ic-3} and the definition \wtraS{\acta} we have that
		\begin{gather}
			\eI{\eSem{\hV}}{\pV}\traS{\actt}^{\!\ast}\eI{\eV'''}{\pV'''} \label{proof:ntt-b-ic-5}\\
			\eI{\eV'''}{\pV'''}\traS{\acta}\eI{\eV''}{\pV''} \label{proof:ntt-b-ic-6}
		\end{gather}
		This information allows us to apply multiple consecutive applications of \Cref{lemma:nvtt-3} on \eqref{proof:ntt-b-ic-2} and \eqref{proof:ntt-b-ic-5} and infer that
		\begin{gather}
			\pV\traS{\actt}^{\ast}\pV''' \label{proof:ntt-b-ic-7} \\
			\forall\trr\cdot\nvsat{\pV'''}{\acta\trr}{\hV} \label{proof:ntt-b-ic-8} \\
			\eV'''=\eSem{\hV}  \label{proof:ntt-b-ic-9}
		\end{gather}
		and by \eqref{proof:ntt-b-ic-6},\eqref{proof:ntt-b-ic-9} and \Cref{lemma:nvtt-4} we have that
		\begin{gather}
			\pV'''\traS{\acta}\pV'' \label{proof:ntt-b-ic-10} \\
			\eV''=\eSem{\afterF{\hV}{\acta}} \label{proof:ntt-b-ic-11}
		\end{gather}
		By \eqref{proof:ntt-b-ic-8}, \eqref{proof:ntt-b-ic-10} and \Cref{lemma:nvtt-1} we know that
		\begin{gather}
			\nvsat{\pV''}{\trr}{\afterF{\hV}{\acta}} \label{proof:ntt-b-ic-12}
		\end{gather}
		With the knowledge of \eqref{proof:ntt-b-ic-4}, \eqref{proof:ntt-b-ic-11} and \eqref{proof:ntt-b-ic-12} we can now apply the \emph{inductive hypothesis} and infer that
		\begin{gather}
			\pV''\wtraS{\trr}\pV'. \label{proof:ntt-b-ic-13}
		\end{gather}
		Finally, by joining together \eqref{proof:ntt-b-ic-7}, \eqref{proof:ntt-b-ic-10} and \eqref{proof:ntt-b-ic-13} with the definition of \wtraS{\tr} we can conclude that
		\begin{gather*}
			\pV\wtraS{\acta\trr}\pV' 
		\end{gather*}
		as required, and so we are done. \vspace{-5mm}
	\end{case}
\end{proof}\smallskip

\paragraph*{Proving \Cref{lemma:nvtt-1}}
\begin{rtp}
	$$\pV\wtraS{\acta}\pV' \text{ and } \nvsat{\pV}{\acta\tr}{\hV} \imp \nvsat{\pV'}{\tr}{\afterF{\hV}{\acta}}$$
	To simplify the proof, we instead prove the contrapositive, \ie
	$$\pV\wtraS{\acta}\pV' \text{ and } \vsat{\pV'}{\tr}{\afterF{\hV}{\acta}} \imp \vsat{\pV}{\acta\tr}{\hV}$$
	 \vspace{-5mm}
\end{rtp}
\setcounter{equation}{0}
\begin{proof} The proof proceeds by rule induction on \afterF{\hV}{\acta}.
	
	\begin{case}[\afterF{\hTru}{\acta}]
		We assume that $\pV\wtraS{\acta}\pV'$ and also that $\vsat{\pV'}{\tr}{\afterF{\hTru}{\acta}}$. This case, however, does not apply since by definition $\afterF{\hTru}{\acta}=\hTru$ which contradicts the assumption that system $\pV'$ and trace \tr violate formula $\afterF{\hTru}{\acta}=\hTru$.
	\end{case}

	\begin{case}[\afterF{\hFls}{\acta}] This case holds \emph{trivially} since by the definition of \vsatL, we know that \hFls is violated regardless of the process or trace, such that we can immediately conclude that 
		\begin{gather*}
			\vsat{\pV}{\acta\tr}{\hFls}
		\end{gather*}
	as required.
	\end{case}

	\begin{case}[\afterF{\hMaxXF}{\acta}]
		We start this case by assuming that
		\begin{gather}
			\pV\wtraS{\acta}\pV' \label{proof:nvtt-1-max-1}\\
			\vsat{\pV'}{\tr}{\afterF{\hMaxXF}{\acta}} \label{proof:nvtt-1-max-2}
		\end{gather}
		Since by definition $\afterF{\hMaxXF}{\acta}=\afterF{\hVMaxXFSub}{\acta}$, by \eqref{proof:nvtt-1-max-1}, \eqref{proof:nvtt-1-max-2} and the \emph{inductive hypothesis} we infer that $\vsat{\pV}{\acta\tr}{\hV\sub{\hMaxXF}{\hVarX}}$, from which by the definition of \vsatL, we can conclude
		\begin{gather*}
			\vsat{\pV}{\acta\tr}{\hMaxXF}
		\end{gather*}
		as required.
	\end{case}

	\begin{case}[\afterF{\hBigAndD{i\in\IndSet}\hNec{\actS_i}\hV_i}{\acta} \text{  when }\exists j{\in}\IndSet\cdot\mtchS{\actS_j}{\acta}{=}\sV]
		We now assume that 
		\begin{gather}
			\pV\wtraS{\acta}\pV' \label{proof:nvtt-1-and-a-1}\\
			\vsat{\pV'}{\tr}{\afterF{\hBigAndD{i\in\IndSet}\hNec{\actS_i}\hV_i}{\acta}}  \label{proof:nvtt-1-and-a-2}\\
			\exists j{\in}\IndSet\cdot\mtchS{\actS_j}{\acta}{=}\sV \label{proof:nvtt-1-and-a-3}
		\end{gather}	
		By \eqref{proof:nvtt-1-and-a-2}, \eqref{proof:nvtt-1-and-a-3} and the definition of \afterFS we deduce that \vsat{\pV'}{\tr}{\hV_j\sV} and subsequently by \eqref{proof:nvtt-1-and-a-1} and the definition of \vsatL we infer that $\vsat{\pV}{\acta\tr}{\hNec{\actS_j}\hV_j}$ upon which by \eqref{proof:nvtt-1-and-a-3} and the definition of \vsatL, we can finally conclude that
		\begin{gather*}
			\vsat{\pV}{\acta\tr}{\hBigAndD{i\in\IndSet}\hNec{\actS_i}\hV_i} 
		\end{gather*}
		as required.
	\end{case}

	\begin{case}[\afterF{\hBigAndD{i\in\IndSet}\hNec{\actS_i}\hV_i}{\acta} \text{  when }\forall i{\in}\IndSet\cdot\mtchS{\actS_i}{\acta}{=}\sVundef]
		Initially we assume that $\pV\wtraS{\acta}\pV'$ and that \vsat{\pV'}{\tr}{\afterF{\hBigAndD{i\in\IndSet}\hNec{\actS_i}\hV_i}{\acta}}. This case, however, does not apply as when $\forall i{\in}\IndSet\cdot\mtchS{\actS_i}{\acta}{=}\sVundef$, then by definition, $\afterF{\hBigAndD{i\in\IndSet}\hNec{\actS_i}\hV_i}{\acta}=\hTru$ which leads to a contradiction since \vsat{\pV'}{\tr}{(\afterF{\hBigAndD{i\in\IndSet}\hNec{\actS_i}\hV_i}{\acta}=\hTru)} is a false assumption by the definition of \vsatL.
	\end{case}
\vspace{-5mm}
\end{proof}

\setcounter{equation}{0}
\paragraph*{Proving \Cref{lemma:nvtt-2}}
\begin{rtp}
	\begin{align*}
		& \quad \nvsat{\pV}{\acta\tr}{\hV} \text{ and } \pV\wtraS{\acta}\pV' \imp \eI{\eSem{\hV}}{\pV}\wtraS{\acta}\eI{\eSem{ \afterF{\hV}{\acta}}}{\pV'} \\
		\equiv & \;\; \exists \hV' \cdot \nvsat{\pV}{\acta\tr}{\hV} \text{ and } \pV\wtraS{\acta}\pV' \text{ and } \afterF{\hV}{\acta}{=}\hV' \imp \eI{\eSem{\hV}}{\pV}\wtraS{\acta}\eI{\eSem{\hV'}}{\pV'}
	\end{align*}\vspace{-3mm}
\end{rtp}

\begin{proof} The proof proceeds by rule induction on \afterF{\hV}{\acta}.
	
	\begin{case}[\afterF{\hTru}{\acta}]
		Initially we assume that: $\afterF{\hTru}{\acta}=\hTru$, $\nvsat{\pV}{\tr}{\hTru}$ and that $\pV\wtraS{\acta}\pV'$ from which we can deduce that
		\begin{gather}
			\pV\wreduc\pV''	 \label{proof:nvtt-2-tt-2a} \\
			\pV''\traS{\acta}\pV'	 \label{proof:nvtt-2-tt-2b}
		\end{gather}
		By applying multiple applications of rule \rtit{iAsy} on \eqref{proof:nvtt-2-tt-2a} we have that
		\begin{gather}
			\eI{\eSem{\hTru}}{\pV}\wreduc\eI{\eSem{\hTru}}{\pV''} \label{proof:nvtt-2-tt-3}
		\end{gather}
		Since $\eSem{\hTru}=\eIden$, by rule \rtit{eId} we have that
		\begin{gather}
			\eSem{\hTru}\traS{\ioact{\acta}{\acta}}\eSem{\hTru} \label{proof:nvtt-2-tt-5}
		\end{gather}
		and hence by \eqref{proof:nvtt-2-tt-2b}, \eqref{proof:nvtt-2-tt-5} and rule \rtit{iTrn} we know that $\eI{\eSem{\hTru}}{\pV''}\traS{\acta}\eI{\eSem{\hTru}}{\pV'}$, and so by \eqref{proof:nvtt-2-tt-3} and transitivity we conclude that
		\begin{gather*}
			\eI{\eSem{\hTru}}{\pV}\wtraS{\acta}\eI{\eSem{\hTru}}{\pV'}
		\end{gather*}
		as required.
	\end{case}

	\begin{case}[\afterF{\hFls}{\acta}]
		Since we assume that $\afterF{\hFls}{\acta}=\hFls$, $\pV\traS{\acta}\pV'$, and that $\nvsat{\pV}{\tr}{\hFls}$, this case does not apply since the last assumption does not hold because the definition of $\vsatL$ states that \hFls is \emph{always violated}.
	\end{case}
	
	\begin{case}[\afterF{\hMaxXF}{\acta}]
		We start by assuming that
		\begin{gather}
			\afterF{\hMaxXF}{\acta}=\afterF{\hV\sub{\hMaxXF}{\hVarX}}{\acta} \label{proof:nvtt-2-max-1}\\
			\pV\wtraS{\acta}\pV'	 \label{proof:nvtt-2-max-2}\\
			\nvsat{\pV}{\acta\tr}{\hMaxXF} \label{proof:nvtt-2-max-3}
		\end{gather}
		From assumption \eqref{proof:nvtt-2-max-1} and by the definition of \afterFS we can deduce that
		\begin{gather}
			\exists\hV'\cdot\afterF{\hV\sub{\hMaxXF}{\hVarX}}{\acta}=\hV' \label{proof:nvtt-2-max-4}
		\end{gather}
		and by applying the definition of $\vsatL$ on assumption \eqref{proof:nvtt-2-max-3}, we infer that
		\begin{gather}
			\nvsat{\pV}{\acta\tr}{\hV\sub{\hMaxXF}{\hVarX}} \label{proof:nvtt-2-max-5}
		\end{gather}
		By knowing \eqref{proof:nvtt-2-max-2}, \eqref{proof:nvtt-2-max-4} and \eqref{proof:nvtt-2-max-5} we can now apply the \emph{inductive hypothesis} and conclude that
		\begin{gather}
			\eI{\eSem{\hV\sub{\hMaxXF}{\hVarX}}}{\pV} \traS{\acta} \eI{\eSem{\hV'}}{\pV'} \label{proof:nvtt-2-max-6}
		\end{gather}
		By \eqref{proof:nvtt-2-max-6} and the definition of \eSem{-}, we know
		\begin{gather}
			\eI{\eSem{\hV}\sub{\rec{\rV}{\eSem{\hV}}}{\rV}}{\pV} \traS{\acta} \eI{\eSem{\hV'}}{\pV'} \label{proof:nvtt-2-max-7}
		\end{gather}
		By \eqref{proof:nvtt-2-max-7} and \rtit{eRec}, we know
		\begin{gather}
			\eI{\rec{\rV}{\eSem{\hV}}}{\pV} \traS{\acta} \eI{\eSem{\hV'}}{\pV'} \label{proof:nvtt-2-max-8}
		\end{gather}
		By \eqref{proof:nvtt-2-max-8} and the definition of \eSem{-}, we know
		\begin{gather*}
			\eI{\eSem{\hMaxXF}}{\pV} \traS{\acta} \eI{\eSem{\hV'}}{\pV'} 
		\end{gather*}
		as required.
	\end{case}
	
	\begin{case}[\afterF{\hBigAndD{i\in\IndSet}\hNec{\actSN{\pate_i}{\bV_i}}\hV_i}{\acta} \text{  when }\exists j{\in}\IndSet\cdot\mtchS{\actS_j}{\acta}{=}\sV]
		We now assume that,
		\begin{gather}
			\afterF{\hBigAndD{i\in\IndSet}\hNec{\actSN{\pate_i}{\bV_i}}\hV_i}{\acta}=\hV_j\sV \label{proof:nvtt-2-and-a-1}
		\end{gather}
		because
		\begin{gather}
			\exists j{\in}\IndSet\cdot\mtchS{\actSN{\pate_j}{\bV_j}}{\acta}{=}\sV \label{proof:nvtt-2-and-a-2}
		\end{gather}
		and
		\begin{gather}
			\pV\wtraS{\acta}\pV'  \label{proof:nvtt-2-and-a-3} \\
			\nvsat{\pV}{\acta\tr}{\hBigAndD{i\in\IndSet}\hNec{\actSN{\pate_i}{\bV_i}}\hV_i} \label{proof:nvtt-2-and-a-4}
		\end{gather}
		Since from \eqref{proof:nvtt-2-and-a-4} we know that process \pV does not violate any of the conjunction branches, and since from \eqref{proof:nvtt-2-and-a-2} we know that system action \acta matches with branch $j$, by the definition of $\vsatL$ we can deduce that $\hV_j\neq\hFls$ (otherwise it would contradict with \eqref{proof:nvtt-2-and-a-4}). This means that by rule \rtit{eTrn} we know that the enforcer will not modify the system action \acta, and so we know that 
		\begin{gather}
			\exists j{\in}\IndSet\cdot \prf{\actSTN{\pate_j}{\bV_j}{\bnd{\pate_j}}}{\eSem{\hV_j}} \traS{\ioact{\acta}{\acta}} \eSem{\hV_j\sV} \label{proof:nvtt-2-and-a-6}
		\end{gather}
		By \eqref{proof:nvtt-2-and-a-6} and \rtit{eSel} we know
		\begin{gather}
			\chBigI\prf{\actSTN{\pate_i}{\bV_i}{\pate'_i}}{\eSem{\hV_i}} \traS{\ioact{\acta}{\acta}} \eSem{\hV_j\sV} \quad (\text{where } \pate'_i{\in}\set{\bnd{\pate_i},\actt}) \label{proof:nvtt-2-and-a-7}
		\end{gather}
		By \eqref{proof:nvtt-2-and-a-7} and \rtit{eRec} we know
		\begin{gather}
			\rec{\rVV}{\chBigI\prf{\actSTN{\pate_i}{\bV_i}{\pate'_i}}{\eSem{\hV_i}}} \traS{\ioact{\acta}{\acta}} \eSem{\hV_j\sV} \quad (\text{where } \pate'_i{\in}\set{\bnd{\pate_i},\actt}) \label{proof:nvtt-2-and-a-8}
		\end{gather}
		By \eqref{proof:nvtt-2-and-a-8} and the definition of \eSem{-} we know
		\begin{gather}
			\eSem{\hBigAndD{i\in\IndSet}\hNec{\actSN{\pate_i}{\bV_i}}\hV_i} \traS{\ioact{\acta}{\acta}} \eSem{\hV_j\sV} \label{proof:nvtt-2-and-a-9}
		\end{gather}
		From \eqref{proof:nvtt-2-and-a-3} and the definition \wtraS{\acta}, we know that $\pV\wreduc\pV''\traS{\acta}\pV'$, which means that by multiple applications of rule \rtit{iAsy} we know that for every enforcer \eV, $\eI{\eV}{\pV}\wreduc\eI{\eV}{\pV''}$, and subsequently by \eqref{proof:nvtt-2-and-a-9} and rule \rtit{iTrn} we infer that		
		\begin{gather*}
			\eI{\eSem{\hBigAndD{i\in\IndSet}\hNec{\actSN{\pate_i}{\bV_i}}\hV_i}}{\pV} \wreduc \eI{\eSem{\hBigAndD{i\in\IndSet}\hNec{\actSN{\pate_i}{\bV_i}}\hV_i}}{\pV''} \traS{\acta}  \eI{\eSem{\hV_j\sV}}{\pV'} 
		\end{gather*}
		as required.
	\end{case}
	
	\begin{case}[\afterF{\hBigAndD{i\in\IndSet}\hNec{\actS_i}\hV_i}{\acta} \text{  when }\forall i{\in}\IndSet\cdot\mtchS{\actS_i}{\acta}{=}\sVundef]
		We start by assuming that
		\begin{gather}
		\afterF{\hBigAndD{i\in\IndSet}\hNec{\actS_i}\hV_i}{\acta}=\hTru \label{proof:nvtt-2-and-b-1}
		\end{gather}
		because
		\begin{gather}
		\forall i{\in}\IndSet\cdot\mtchS{\actS_i}{\acta}{=}\sVundef \label{proof:nvtt-2-and-b-2}
		\end{gather}
		and
		\begin{gather}
			\pV\wtraS{\acta}\pV'  \label{proof:nvtt-2-and-b-3} \\
			\nvsat{\pV}{\acta\tr}{\hBigAndD{i\in\IndSet}\hNec{\actSN{\pate_i}{\bV_i}}\hV_i} \label{proof:nvtt-2-and-b-4}
		\end{gather}
		By \eqref{proof:nvtt-2-and-b-2} and the definition of $\vsatL$ we know
		\begin{gather}
			\forall i{\in}\IndSet\cdot \prf{\actSTN{\pate_i}{\bV_i}{\pate'_i}}{\eSem{\hV_i}} \ntraS{\acta} \quad (\text{where } \pate'_i{\in}\set{\bnd{\pate_i},\actt}) \label{proof:nvtt-2-and-b-5}
		\end{gather}
		By \eqref{proof:nvtt-2-and-b-5} and \rtit{eSel} we know
		\begin{gather}
			\chBigI\prf{\actSTN{\pate_i}{\bV_i}{\pate'_i}}{\eSem{\hV_i}} \ntraS{\acta} \quad (\text{where } \pate'_i{\in}\set{\bnd{\pate_i},\actt}) \label{proof:nvtt-2-and-b-6}
		\end{gather}
		By \eqref{proof:nvtt-2-and-b-6} and \rtit{eRec} we know
		\begin{gather}
			\rec{\rVV}{\chBigI\prf{\actSTN{\pate_i}{\bV_i}{\pate'_i}}{\eSem{\hV_i}}} \ntraS{\acta} \quad (\text{where } \pate'_i{\in}\set{\bnd{\pate_i},\actt}) \label{proof:nvtt-2-and-b-7}
		\end{gather}
		By \eqref{proof:nvtt-2-and-b-7} and the definition of \eSem{-} we know
		\begin{gather}
			\eSem{\hBigAndD{i\in\IndSet}\hNec{\actSN{\pate_i}{\bV_i}}\hV_i} \ntraS{\acta} \label{proof:nvtt-2-and-b-8}
		\end{gather}
		From \eqref{proof:nvtt-2-and-b-3} and the definition \wtraS{\acta}, we know that $\pV\wreduc\pV''\traS{\acta}\pV'$, which means that by multiple applications of rule \rtit{iAsy} we know that for every enforcer \eV, $\eI{\eV}{\pV}\wreduc\eI{\eV}{\pV''}$, and subsequently by \eqref{proof:nvtt-2-and-b-8} and rule \rtit{iTer} we infer that		
		\begin{gather}
		\eI{\eSem{\hBigAndD{i\in\IndSet}\hNec{\actSN{\pate_i}{\bV_i}}\hV_i}}{\pV} \wreduc \eI{\eSem{\hBigAndD{i\in\IndSet}\hNec{\actSN{\pate_i}{\bV_i}}\hV_i}}{\pV''} \traS{\acta} \eI{\eIden}{\pV'} \label{proof:nvtt-2-and-b-9}
		\end{gather}
		Finally by \eqref{proof:nvtt-2-and-b-9} and the definitions of \eSem{-} and \wtraS{\acta} we conclude that
		\begin{gather*}
			\eI{\eSem{\hBigAndD{i\in\IndSet}\hNec{\actSN{\pate_i}{\bV_i}}\hV_i}}{\pV} \wtraS{\acta} \eI{\eSem{\hTru}}{\pV'}
		\end{gather*}
		as required and so we are done.
	\end{case} \vspace{-5mm}
	
\end{proof}

\setcounter{equation}{0}
\paragraph*{Proving \Cref{lemma:nvtt-3}}
\begin{rtp}
	\begin{align*}
		\nvsat{\pV}{\tr}{\hV} \text{ and } \eI{\eSem{\hV}}{\pV}\traS{\actt}\eI{\eV'}{\pV'} \imp \pV\traS{\actt}\pV' \text{ and } \eV'=\eSem{\hV} \text{ and } \nvsat{\pV'}{\tr}{\hV}
	\end{align*}\vspace{-3mm}
\end{rtp}

\begin{proof} The proof proceeds by rule induction on $\eI{\eSem{\hV}}{\pV}\traS{\actt}\eI{\eV'}{\pV'}$.
	
	\begin{Cases}[\text{\rtit{iTer} and \rtit{iIns}}] These cases do not apply as \rtit{iTer} only transitions over visible actions \acta, while \rtit{iIns} cannot be applied as \eSem{\hV} does not synthesise insertion monitors.
	\end{Cases}

	\begin{case}[\rtit{iAsy}] We assume that 
		\begin{gather}
			\forall \tr\cdot\nvsat{\pV}{\tr}{\hV}  \label{proof-3-nvtt-1}\\
			\eI{\eSem{\hV}}{\pV}\traS{\actt}\eI{\eV'}{\pV'} \label{proof-3-nvtt-2}
		\end{gather}
		because
		\begin{gather}
			\pV\traS{\actt}\pV' \label{proof-3-nvtt-3} \\
			\eV'=\eSem{\hV}  \label{proof-3-nvtt-4}
		\end{gather}
		Since the violation semantics are agnostic of \actt-actions, from \eqref{proof-3-nvtt-1} and \eqref{proof-3-nvtt-3} we can deduce that 
		\begin{gather}
			\forall \tr\cdot\nvsat{\pV'}{\tr}{\hV}  \label{proof-3-nvtt-5}
		\end{gather}
		and so we are done by \eqref{proof-3-nvtt-3}, \eqref{proof-3-nvtt-4} and \eqref{proof-3-nvtt-5}.
	\end{case}

	\begin{case}[\rtit{iTrn}] We assume that 
		\begin{gather}
		\forall \tr\cdot\nvsat{\pV}{\tr}{\hV}  \label{proof-3-nvtt-itrn-1}\\
		\eI{\eSem{\hV}}{\pV}\traS{\actt}\eI{\eV'}{\pV'} \label{proof-3-nvtt-itrn-2}
		\end{gather}
		because
		\begin{gather}
		\pV\traS{\acta}\pV' \label{proof-3-nvtt-itrn-3} \\
		\eSem{\hV}\traS{\ioact{\acta}{\actt}}\eV'  \label{proof-3-nvtt-itrn-4}
		\end{gather}
		By the rules in our model we know that the suppressing transition in \eqref{proof-3-nvtt-itrn-4} can only take place if the monitor is capable of performing the suppression transformation, \ie has the form $\rec{\rVV}{\ch{\prf{\actSTD{\pate_j}{\bV_j}}{\rVV}}{ \chBig{i\in\IndSet\setminus\set{j}}{\begin{xbrace}{cl}
				\prf{\actSTD{\pate_j}{\bV_j}}{\rVV} & \text{  if }\hV_i=\hFls\\
				\prf{\actSN{\pate_j}{\bV_j}}{\eSem{\hV_i}} & \text{  otherwise}
			\end{xbrace}  }}}$ where $\exists j\cdot\mtchS{\prf{\actSTD{\pate_j}{\bV_j}}}{\acta}=\sV$. By the definition of \eSem{\hV} this monitor can only be synthesised if \hV has the form of $\hAnd{\hNec{\actSN{\pate_j}{\bV_j}}\hFls}{ \hBigAndU{i\in\IndSet\setminus\set{j}}\hNec{\actSN{\pV_i}{\bV_i}}\hV_i }$, which means that every \acta prefixed trace would violate \hV since when \acta satisfies the conjunct necessity $\hNec{\actSN{\pate_j}{\bV_j}}\hFls$ every suffix \trr of the trace would violate \hFls and so we have that $\forall\trr\cdot\vsat{\pV}{\acta\trr}{\hV}$. This therefore contradicts with assumption \eqref{proof-3-nvtt-itrn-1} and hence this case does not apply. \vspace{-3mm}
	\end{case}

\end{proof}

\setcounter{equation}{0}
\paragraph*{Proving \Cref{lemma:nvtt-4}}
\begin{rtp}
	\begin{align*}
	 \eI{\eSem{\hV}}{\pV}\traS{\acta}\eI{\eV'}{\pV'} \imp \pV\wtraS{\acta}\pV' \text{ and } \eV'=\eSem{\afterF{\hV}{\acta}} 
	\end{align*}\vspace{-3mm}
\end{rtp}

\begin{proof} The proof proceeds by rule induction on $\eI{\eSem{\hV}}{\pV}\traS{\acta}\eI{\eV'}{\pV'}$.
	
	\begin{Cases}[\text{\rtit{iAsy} and \rtit{iIns}}] These cases do not apply since \rtit{iAsy} transitions over \actt actions only, while \rtit{iIns} cannot ever be applied since \eSem{\hV} does not synthesise insertion monitors.
	\end{Cases}
		
		\begin{case}[\rtit{iTer}] We assume that $\eI{\eSem{\hV}}{\pV}\traS{\acta}\eI{\eIden}{\pV'}$ because 
		\begin{gather}
		\pV\tra{\acta}\pV' \label{proof-4-nvtt-iter-2} \\
		\eSem{\hV}\ntra{\acta}\,\land\,\eSem{\hV}\ntra{\actdot} \label{proof-4-nvtt-iter-3}
		\end{gather}
		From the definition of \eSem{\hV} and the rules in our model we know that \eqref{proof-4-nvtt-iter-3} is only possible when $\hV=\hBigAndU{i\in\IndSet}{\hNec{\actSN{\pate_i}{\bV_i}}\hV_i}$ and $\forall i\cdot\mtchS{\actSN{\pate_i}{\bV_i}}{\acta}=\sVundef$ as this would be synthesised into an enforcer of the $\eV=\rec{\rVV}{\chBigI\prf{\actSTN{\pate_i}{\bV_i}{\pate_i'}}{\eV'}}$ where every branch is unable to match with \acta. Knowing that \hV can only have this form and by the definition of \afterFS we deduce that 
		\begin{gather}
			\afterF{\hV}{\acta}=\hTru \label{proof-4-nvtt-iter-4}
		\end{gather}
		Since by the definition of \eSem{-} we know that $\eIden=\eSem{\hTru}$, by \eqref{proof-4-nvtt-iter-4} we can conclude that
		\begin{gather}
			\eIden=\eSem{\afterF{\hV}{\acta}} \label{proof-4-nvtt-iter-5}
		\end{gather}
		and hence this case is done by \eqref{proof-4-nvtt-iter-2} and \eqref{proof-4-nvtt-iter-5}.
	\end{case}

	\begin{case}[\rtit{iTrn}] We assume that $\eI{\eSem{\hV}}{\pV}\traS{\acta}\eI{\eV'}{\pV'}$ because 
		\begin{gather}
			\pV\tra{\actb}\pV' \label{proof-4-nvtt-itrn-2} \\
			\eSem{\hV}\tra{\ioact{\actb}{\acta}}\eV' \label{proof-4-nvtt-itrn-3}
		\end{gather}
		From the definition of \eSem{\hV} we can infer that our synthesis cannot generate action replacing monitors and hence we can deduce that
		\begin{gather}
			\acta=\actb \label{proof-4-nvtt-itrn-4}
		\end{gather}
		From the definition of \eSem{\hV} and the rules in our model we can also deduce that when $\acta=\actb$ (as confirmed by \eqref{proof-4-nvtt-itrn-4}), \eqref{proof-4-nvtt-itrn-3} occurs only when $\hV=\hBigAndU{i\in\IndSet}{\hNec{\actSN{\pate_i}{\bV_i}}\hV_i}$ and $\exists j\cdot\mtchS{\actSN{\pate_j}{\bV_j}}{\acta}=\sV$ as this would be synthesised into an enforcer of the form
		\begin{gather}
			\eSem{\hV}=\rec{\rVV}{\ch{\prf{\actSN{\pate_j}{\bV_j}}{\eSem{\hV_j}}}{ \chBig{i\in\IndSet\setminus\set{j}}\begin{xbrace}{rl}
						\prf{\actSTD{\pate_i}{\bV_i}}{\rVV} & \;\;(\text{if }\hV_i=\hFls)\\
						\prf{\actSN{\pate_i}{\bV_i}}{\eSem{\hV_i}} & \;\;(\text{otherwise})
					\end{xbrace} }} \label{proof-4-nvtt-itrn-4.5}
		\end{gather}
		where only the branch with index $j$ can match \acta. Knowing that \hV can only have this form, by the definition of \afterFS we can deduce that 
		\begin{gather}
			\afterF{\hV}{\acta}=\hV_j\sV \label{proof-4-nvtt-itrn-5}
		\end{gather}
		Hence, by applying rules \rtit{eRec}, \rtit{eSel} and \rtit{eTrn} on \eqref{proof-4-nvtt-itrn-3}, with the knowledge of \eqref{proof-4-nvtt-itrn-4.5} we know that $\eV'=\eSem{\hV_j\sV}$ and hence by \eqref{proof-4-nvtt-itrn-5} we can conclude that 
		\begin{gather}
			\eV'=\eSem{\afterF{\hV}{\acta}} \label{proof-4-nvtt-itrn-6}
		\end{gather}
		and so we are done by \eqref{proof-4-nvtt-itrn-2}, \eqref{proof-4-nvtt-itrn-4} and \eqref{proof-4-nvtt-itrn-6}.
	\end{case}

\end{proof}

\end{document}